\pdfoutput=1
\RequirePackage{ifpdf}
\ifpdf 
\documentclass[pdftex]{sigma}
\else
\documentclass{sigma}
\fi

\numberwithin{equation}{section}

\usepackage{tikz}

\newtheorem{Theorem}{Theorem}[section]
\newtheorem*{Theorem*}{Theorem}

\newtheorem{Lemma}[Theorem]{Lemma}
\newtheorem{Proposition}[Theorem]{Proposition}
 { \theoremstyle{definition}
\newtheorem{Definition}[Theorem]{Definition}

 }

\begin{document}

\allowdisplaybreaks

\newcommand{\arXivNumber}{2402.08944}

\renewcommand{\PaperNumber}{085}

\FirstPageHeading

\ShortArticleName{The Racah Algebra of Rank 2: Properties, Symmetries and Representation}

\ArticleName{The Racah Algebra of Rank 2:\\ Properties, Symmetries and Representation}

\Author{Sarah POST and S\'{e}bastien BERTRAND}
\AuthorNameForHeading{S.~Post and S.~Bertrand}
\Address{Department of Mathematics, University of Hawai`i at M\={a}noa, Honolulu, Hawai`i, USA}
\Email{\href{mailto:spost@hawaii.edu}{spost@hawaii.edu}, \href{mailto:sbertran@hawaii.edu}{sbertran@hawaii.edu}}

\ArticleDates{Received March 08, 2024, in final form September 10, 2024; Published online September 22, 2024}

\Abstract{The goals of this paper are threefold. First, we provide a new ``universal'' definition for the Racah algebra of rank 2 as an extension of the rank-1 Racah algebra where the generators are indexed by subsets and any three disjoint indexing sets define a subalgebra isomorphic to the rank-1 case. With this definition, we explore some of the properties of the algebra including verifying that these natural assumptions are equivalent to other defining relations in the literature. Second, we look at the symmetries of the generators of the rank-2 Racah algebra. Those symmetries allows us to partially make abstraction of the choice of the generators and write relations and properties in a different format. Last, we provide a novel representation of the Racah algebra. This new representation requires only one generator to be diagonal and is based on an expansion of the split basis representation from the rank-1 Racah algebra.}


\Classification{81R12; 20C35}

\section{Introduction}

The Racah algebra and its generalization has recently been the subject of interest in representation theory and in connection with (multivariate) orthogonal polynomials, see for example the recent review article \cite{Vinet2020} and references therein. The title of the algebra refers to Giulio Racah, who defined what are now known as Racah coefficients to express the recoupling of triples of angular momentum \cite{Racah}. The terminology is commonly used now to express recoupling coefficients of triples of various irreducible representations, including of quantum groups \cite{Kirillov, Levy,Louck}. Back in the Istanbul Summer School \cite{Racah2}, Racah even proposed an algebra that was closing of degree~3, compared to the Racah algebra that closes quadratically. Later the Racah polynomials, in particular their quantum versions, were defined by Askey and Wilson explicitly as {\it ``A set of orthogonal polynomials that generalize the Racah coefficients or $6$-j symbols''} \cite{Askey}.

As is well known thanks to the work of Zhedanov, bispectral orthogonal polynomials admit quadratic algebras \cite{GLZ} and the Racah algebra is the algebra associated with the classical ${(q=1)}$ Racah polynomials. Again, it was Zhedanov \cite{GZ88} who drew an explicit connection between the coefficients of Racah and the symmetry algebra. Now we understand the Racah algebra and the Racah polynomials in particular in terms of tensor products of $\mathfrak{sl}_2(\mathbb{C})$ and its real form~$\mathfrak{su}(1,1)$~\cite{GVZ2013,GZ93, SP}. Due to the inductive nature of the tensor product, this construction naturally generalizes to higher rank \cite{BGVV2018,Vinet2020} and can be seen as the commutant of the diagonal embedding of $\mathfrak{sl}_2(\mathbb{C})$ into the $n$-fold tensor product of its universal algebra \cite{CGPAV22}.

From a different perspective, there has been work to define a ``universal'' version of the Racah algebra \cite{BCH2020,CGPAV22}, divorced from its realization in terms of tensor products, and to understand its representation theory \cite{HBC2020}. This paper follows in this vein. In particular, we define the rank-2 Racah algebra $R(4)$ as an associative algebra with a set of generators and relations and show that it is equivalent to the definition given in \cite{CGPAV22}. In addition, we use the Jacobi identity to recover relation (3.9) of \cite{CGPAV22} and show that the Jacobi identity for all commutation relations is satisfied modulo this relation.

As with the motivation for the rank-1 case, we are interested in understanding the representation theory of this algebra. To this end, we discuss possible choice of linearly independent generator and give a novel representation based on the split basis representation of the rank-1 algebra \cite{HBC2020}.

The manuscript is organized as follows. In each section of the core of this manuscript, we first review some characteristics of the rank-1 Racah algebra before tackling the rank-2 case.
In Section \ref{defn}, we propose a new universal definition of the rank-2 Racah algebra as an extension of the rank-1 case, and we explore its properties and relations. In Section \ref{gen}, we consider a special set of generators to construct the rank-2 Racah algebra, the contiguous basis. The goal of this section is to partially make abstraction of the choice of generator using discrete symmetries. In addition, we use these symmetries to further explore the properties of the Racah algebra of rank~2. In the last core section, Section \ref{repn}, we provide the assumptions of how we obtained a~new representation. The explicit coefficients of the representation can be found in Appendix \ref{AppCoeff}, and some proofs from Section \ref{defn} are provided in Appendix \ref{AppendixA}.

\section{Definition and properties of Racah algebras}
\label{defn}
\subsection{The rank-1 Racah algebra}
\begin{Definition}
The Racah algebra of rank~$1$, denoted $R(3)$ is defined to be the unital, associative algebra over a field $\mathcal{K}$ with generators $C_{I}$ with $I \subset \{ 1,2,3\}$ satisfying the following relations:
\begin{gather}
[C_{i},C_{I}] = 0 \quad \mbox{for all}\quad i=1,2,3, \qquad [C_{123},C_{I}] = 0,\label{central1}\\
C_{123}=C_{12}+C_{23}+C_{13}-C_1 -C_2-C_3, \label{R3decom}\\
\tfrac{1}{2} [C_{jk},[C_{ij}, C_{jk}]]= C_{ik}C_{jk}-C_{jk}C_{ij}+(C_k-C_j)(C_i-C_{ijk}), \qquad i\neq j \neq k\neq i,\label{quad3}
\end{gather}
where $[\cdot,\cdot]$ is the regular commutator.
\end{Definition}

 To simplify the notation, we omit the set notation and write $C_{ij} = C_{\{ i,j\}}$, although it is important to remember that the generators $C_I$ are symmetric in the indices. Indeed, it is worth noting that the set of equations \eqref{central1}--\eqref{quad3} are invariant under the group of permutation $\mathcal{P}_3$ acting on the indices. In addition, the maximal elements $C_{123}$ and the elements with one index, $C_i$, $i=1,2,3$, are in the center of the algebra. Furthermore, one should note that the Racah algebra is not a Lie algebra as it closes quadratically.

Often, an additional linearly independent element is introduced to avoid commutators of commutators, that is,
\begin{equation}
D_{123}\equiv \tfrac 12 [C_{12}, C_{23}],
\end{equation}
which satisfies
\begin{equation}
D_{123}= \tfrac12 [C_{23}, C_{13}]= \tfrac12 [C_{13}, C_{12}],\label{D123alt}
\end{equation}
as a consequence of relations \eqref{central1}--\eqref{R3decom}.
With this definition, we can express \eqref{quad3} as
\begin{equation}
[C_{jk}, D_{ijk}] =C_{ik}C_{jk}-C_{jk}C_{ij}+(C_k-C_j)(C_i-C_{ijk}),\label{rank1inner}
\end{equation}
and cyclic permutations thereof. Note that unlike the $C$ generators, the element $D_{123}$ is not fully invariant under $\mathcal{P}_3$. Non-cyclic permutations of the indices 1, 2, 3 change the sign of the element $D_{123}$.

With this definition, it is possible to chose a linearly independent basis of generators, for example including the central elements $C_1$, $C_2$, $C_3$, $C_{123}$ and the non-central ones $C_{12}$, $C_{23}$, and~$D_{123}$.

\begin{Lemma}\label{JacobiRank1}
In the algebra $R(3)$, the Jacobi identity is satisfied without additional constraints.
\end{Lemma}
\begin{proof}
 We need to check only the non-central terms $C_{12}$, $ C_{23}$, $C_{13}$ and $D_{123}$. The first case to check is
\begin{gather*}
[C_{13},[C_{12}, C_{23}]]+ [C_{12},[C_{23}, C_{13}]]+[C_{23},[C_{13}, C_{12}]]\\
\qquad= 2 [C_{13}, D_{123}]+2[C_{12}, D_{123}]+2[C_{23}, D_{123}]\\
\qquad= 2 \left( C_{12}C_{13}-C_{13}C_{23}+(C_1-C_3)(C_2-C_{123}) + C_{23}C_{12}-C_{12}C_{13}\right.\\
\phantom{=}{}\qquad\left.+(C_2-C_1)(C_3-C_{123}) +C_{13}C_{23}-C_{23}C_{12}+(C_3-C_2)(C_1-C_{123})\right) = 0.
\end{gather*}
Here, we have used \eqref{rank1inner} and the symmetry of $D_{123}$ \eqref{D123alt}.

The second (and last) case to check is a pair of $C$'s with the $D$. Up to permutation, it is as follows:
\begin{gather*}
[C_{12},[C_{23}, D_{123}]]+ [C_{23},[D_{123}, C_{12}]]+[D_{123},[C_{12}, C_{23}]]\\
\qquad= [ C_{12},C_{13}C_{23}-C_{23}C_{12}]- [C_{23}, C_{23}C_{12}-C_{12}C_{13}]+0\\
\qquad= -D_{123}C_{23}+C_{13}D_{123}- D_{123} C_{12} +C_{23}D_{123} -D_{123} C_{13}+C_{12}D_{123}  \\
\qquad= [C_{13}, D_{123}]+[C_{12}, D_{123}]+[C_{23}, D_{123}] =0.
\end{gather*}
We have shown in the rank-1 Racah algebra, the commutation relations satisfy the Jacobi identity without further constraints.
\end{proof}

 We finish our discussion of the rank-1 case by noting that the algebra possesses a Casimir operator, that is,
\begin{gather}
\mathfrak{C}= D_{123}^2-\tfrac{1}{2}\big\lbrace C_{12}^2,C_{23}\big\rbrace-\tfrac{1}{2}\big\lbrace C_{23}^2,C_{12}\big\rbrace+C_{12}^2+C_{23}^2+\lbrace C_{12},C_{23}\rbrace\nonumber\\
\hphantom{\mathfrak{C}=}{}+\tfrac{1}{2}(C_1+C_2+C_3+C_{123})(\lbrace C_{12},C_{23}\rbrace-2C_{12}-2C_{23})\nonumber\\
\hphantom{\mathfrak{C}=}{}+(C_2-C_3)(C_{123}-C_1)C_{12}+(C_2-C_1)(C_{123}-C_3)C_{23}\nonumber\\
\hphantom{\mathfrak{C}=}{}+(C_1+C_3)(C_{123}+C_2)+(C_1C_3-C_{123}C_2)(C_{123}-C_1+C_2-C_3),\label{Casmir1}
\end{gather}
where the last line involves only central elements. If one uses the $\mathfrak{su}(1,1)$ approach to construct a~representation of a Racah algebra, see, e.g., \cite{GVZ2013}, then it has been shown that the Casimir operator will take a value of 0, which corresponds to the so-called special Racah. However, in this paper, we will not start from $\mathfrak{su}(1,1)$, but from a definition of the algebra, both for the rank-1 and rank-2 cases.

\subsection{The rank-2 and higher Racah algebras}

For the higher-rank extension, we extend the indexing set in the natural manner and replace the decomposition \eqref{R3decom} and quadratic relation \eqref{quad3} with multi-indices. The commutation relations are extended in a manner that agrees with the coupling structure of the tensor product. That is, for any triple of disjoint subsets $I$, $J$, $K$, the subalgebra generated will be isomorphic to the rank-1 case. This would be analogous to reducing these factors into irreducible components and then taking their tensor product.

\begin{Definition}
The Racah algebra $R(n)$ is defined to be the unital associative algebra over a~field~$\mathcal{K}$ with generators $C_{I}$ with $I \subset \{ 1,2,3, \dots, {n}\}$ satisfying the following relations:
\begin{gather}
[C_{I},C_{J}] = 0\quad \mbox{for all}\quad I \subset J\quad \mbox{or}\quad I \cap J = \varnothing, \label{central2}\\
C_{IJK}=C_{IJ}+C_{JK}+C_{IK}-C_I -C_J-C_K, \label{R4decom}\\
\tfrac12 [C_{JK},[C_{IJ}, C_{JK}]]= C_{IK}C_{JK}-C_{JK}C_{IJ}+(C_K-C_J)(C_I-C_{IJK}) \label{quad4},
\end{gather}
where $I$, $J$ and $K$ are disjoint subsets in equations \eqref{R4decom}, \eqref{quad4}. For notational convenience, we omit the set notation and union sign and express $IJ=I \cup J$.
\end{Definition}
\begin{Lemma} \label{permutationsymofDs}
The following relations hold in the algebra $R(n)$:
\[
[C_{IJ}, C_{JK}]=[C_{KI}, C_{IJ}] =[C_{JK}, C_{KI}].
\]
\end{Lemma}
\begin{proof}
 This follows directly from the commutation \eqref{central2} and decomposition relations \eqref{R4decom}.
 \end{proof}
\begin{Lemma} The following relation holds in the algebra $R(4)$:
\begin{equation}
\tfrac12 [C_{KI},[C_{IJ}, C_{JK}]]= C_{IJ}C_{KI}-C_{KI}C_{JK}+(C_I-C_K)(C_J-C_{IJK}). \label{quad4b}
\end{equation}
\end{Lemma}
\begin{proof}
 From \eqref{quad4} with $K\rightarrow I \rightarrow J$, we have
\[\tfrac12 [C_{KI},[C_{JK}, C_{KI}]]= C_{IJ}C_{KI}-C_{KI}C_{JK}+(C_I-C_K)(C_J-C_{IJK}).\]
The resulting identity follows from the relations in Lemma \ref{permutationsymofDs}.
\end{proof}

\begin{Lemma} \label{JacobiC} The quadratic relations \eqref{quad4} and \eqref{quad4b} satisfy the Jacobi identity.
\end{Lemma}
\begin{proof} This proof is directly analogous to Lemma \ref{JacobiRank1} with indices replaced by subset. Relation~\eqref{quad4} is a straightforward computation that follows from symmetry of the commutator
\[ [C_{JK},[C_{IJ}, C_{JK}]]+ [C_{IJ},[C_{JK}, C_{JK}]]+[C_{JK},[C_{JK},C_{IJ}]]=0.\]
Relation \eqref{quad4b} is less obvious. We check
\begin{gather*}
[C_{KI},[C_{IJ}, C_{JK}]]+ [C_{IJ},[C_{JK}, C_{KI}]]+[C_{JK},[C_{KI}, C_{IJ}]]\\
\qquad= 2  ( C_{IJ}C_{KI}-C_{KI}C_{JK}+(C_I-C_K)(C_J-C_{IJK}) + C_{JK}C_{IJ}-C_{IJ}C_{KI}  \\
\phantom{\qquad=}{}   +(C_J-C_I)(C_K-C_{IJK})+C_{KI}C_{JK}-C_{JK}C_{IJ}+(C_K-C_J)(C_I-C_{IJK}) ) = 0.
\end{gather*}
Thus the triple commutation of any of the generators satisfies the Jacobi identity identically, as expected in an associative algebra.
\end{proof}

As in the rank-one case, it is often useful to include additional generators so that the quadratic relations \eqref{quad4} become commutators. These new elements are
\[
 D_{ijk}\equiv\tfrac12 [C_{ij}, C_{jk}], \qquad 1\leq i< j< k \leq n.
\]
Note that unlike the $C$'s, these elements are defined for a trio of distinct \emph{ordered} indices. However, we see that they are invariant under cyclic permutations as a result of Lemma~\ref{permutationsymofDs}
\[
D_{ijk}=\tfrac12 [C_{jk}, C_{ik}]=\tfrac12 [C_{ik}, C_{ij}].
\]

As has been observed by others, it can simplify expressions to introduce new shifted generators{\samepage
\begin{gather}
P_{ij} =C_{ij}-C_i-C_j, \qquad i,j\in \{ 1,2,3,4\}.\label{P's}
\end{gather}
Note that with this definition the choice of repeated subscript $P_{ii} = -C_{i}$.}

With these definitions, we can express the quadratic relations for singleton subsets as
\begin{equation}
 [P_{jk},D_{ijk}]= P_{jk}P_{ki}-P_{ij}P_{jk}-2P_{j}P_{ki} +2P_{ij}P_k. \label{intquad}
 \end{equation}
We will sometimes refer to these as interior quadratic relations because they agree with the rank-one quadratic relations restricted to the subset $\{ i,j,k\}$. However, they do not span the full set of commutation relations. In particular, they do not include commutators of the form~$[P_{k\ell}, D_{ijk}]$ for distinct choices of $i$, $j$, $k$, $\ell$. These can be obtained by taking appropriate linear combinations of~\eqref{quad4} and higher-order implications of the commutation \eqref{central2} and decomposition relations~\eqref{R4decom}.

For example, we can obtain the following identity using the commutation \eqref{central2} and decomposition relations \eqref{R4decom}.
\begin{Lemma}\label{PDcomm}
The following commutation relations hold in $R(n)$:
\[
  [C_{k\ell}, D_{ijk}] +[C_{k\ell} ,D_{ij\ell}]=0.
  \]
\end{Lemma}
\begin{proof}
 Note that $C_{k\ell}$ and hence $P_{k\ell}$ will commute with both $C_{ik\ell}$ and $C_{jk\ell}$. We use the decomposition relation \eqref{R4decom} to obtain
\begin{align*}
[C_{ik\ell}, C_{jk\ell}]&= [C_{ik}+C_{k\ell}+C_{i\ell}, C_{jk} +C_{k\ell}+C_{j\ell}]\\
&= [C_{ik},C_{kj}]+[C_{ik}, C_{k\ell}] + [C_{\ell k}, C_{kj}] + [C_{k\ell}, C_{\ell j}] + [C_{i\ell}, C_{\ell k} ]+[C_{i \ell}, C_{\ell j}]\\
&= 2 ( D_{ikj} + D_{ik\ell} +D_{\ell k j} + D_{k\ell j} +D_{i\ell k} +D_{i\ell j} ).
\end{align*}
Next we use the fact that $D_{ijk}$ is invariant under cyclic permutations and odd under singleton flips $D_{ijk}=-D_{ikj}$. This leaves
$
[C_{ik\ell}, C_{jk\ell}]=-2D_{ijk} -2D_{ij\ell}$.
So, finally we see that Lemma~\ref{PDcomm} holds from the commutativity of $C_{k\ell}$ with $C_{ik\ell}$ and $C_{jk\ell}$.
\end{proof}

Of course, this lemma does not give us the required relation $[C_{k\ell}, D_{ijk}]$ rather relations between such terms. For example in $R(4)$, we are looking to obtain 24 equation with leading order terms $[P_{A}, D_{B}]$ with $A$ a two element subset and $B$ a three element subset. We have 12 ``inner'' equations of the form \eqref{intquad}. There are 6 equations coming from Lemma \ref{PDcomm} plus 4 from the fact that $C_{ijk\ell}$ is central and 4 from the fact that $C_{ijk}$ commutes with $D_{ijk}$. Finally, there are 12 additional choices of $I$, $J$, $K$ where at least one subset contains 3 elements. This leads to consistent linear system that can be solved for the required relations.

\begin{Lemma}\label{outerlemma}
In the Racah algebra $R(n)$, the elements $P_{ij}$ and $D_{jk\ell}$ satisfy the following relation:
\begin{gather}
[P_{ij}, D_{jk\ell}] =P_{i\ell}P_{jk}-P_{j\ell}P_{ik} .\label{outter}
\end{gather}
\end{Lemma}

\begin{proof} As described above, we used a symbolic manipulator (\textsc{Maple} \cite{Maple}) to solve the linear system of equations involving terms $P_{A} D_B$ for the required expressions. The resulting identities are non-trivial but can be checked as follows. On the one hand, it can be directly verified using the commutation \eqref{central2} and decomposition relations \eqref{R4decom} that $[C_{ij}, D_{ijk}]$ can be expressed as
\begin{gather*}
[C_{ij},D_{jkl}]= - 4[C_{ijk\ell}-C_{ik\ell}, D_{ik\ell}]+2[C_{ijk\ell}-C_{jk\ell}, D_{jk\ell}]
+2[C_{ijk\ell}-C_{ijk}, D_{ijk}] \\
\hphantom{[C_{ij},D_{jkl}]=}{}- 2[C_{ijk\ell}-C_{ij\ell}, D_{ij\ell}]+4[C_{ij}, D_{jk\ell}+D_{ik\ell}]- 2[C_{ik},D_{jkl}-D_{ij\ell}]
\\
\hphantom{[C_{ij},D_{jkl}]=}{}- 2[C_{i\ell}, D_{jk\ell}+D_{ijk}]
-4 [C_{k\ell}, D_{ik\ell}]-2 [C_{j\ell}, D_{ij\ell}]
+2[C_{jk}, D_{ijk}]\\
\hphantom{[C_{ij},D_{jkl}]=}{}+4[C_{jk\ell}, D_{ik\ell}-D_{ijk}]+2[C_{jk\ell}, D_{ijk}+D_{ij\ell}].
\end{gather*}
 On the other hand, the first 2 lines are equal to 0 as consequences of the commutation relations~\eqref{central2}, as $C_{ijkl}$ and $C_{ijk}$ will both commute with $D_{ijk}$ (and permutations). The third line also vanishes from the relation in Lemma \ref{PDcomm}. Finally, the second to last line is comprised of terms containing ``inner'' quadratic relations for singleton subsets \eqref{intquad} and the last line ``inner'' quadratic relations for multi-index sets $I=\{i\}$, $J=\{j\}$, $K=\{k\ell\}$ and $I=\{i\}$, $J=\{k\}$, $ K=\{j\ell\}$ respectively in \eqref{quad4}. Replacing these last two lines with their quadratic identities leaves the resulting equation
 \[ [C_{ij}, D_{jk\ell}] =(C_{i\ell}- C_i -C_j) (C_{jk}-C_j-C_k)-(C_{j\ell}-C_j-C_{\ell})(C_{ik}-C_i-C_k),\]
which can be more succinctly expressed using the $P$ generators as in \eqref{outter}.
\end{proof}

\begin{Lemma} \label{allPDrel|}
The outer commutators satisfy the following identities:
\begin{gather*}
 [P_{ij}, D_{jk\ell}] =- [P_{ji} ,D_{ik\ell}],
\qquad
[P_{ij}, D_{jk\ell}]= [P_{k\ell}, D_{\ell ij}]
\end{gather*}
and
\begin{equation}
[P_{ij},D_{jk\ell}]+[P_{kj}, D_{j\ell i}]+[P_{\ell j}, D_{j ik}]=0. \label{switchPD}
\end{equation}
\end{Lemma}
\begin{proof}
 Each of these are direct consequences of Lemma \ref{outerlemma}. We also note that the first equation is exactly Lemma \ref{PDcomm} expressed in terms of $P$.
 \end{proof}

For the final relations, we would like to understand the commutation of the $D$'s with each other. These can be determined from the definition of the $D$'s plus the assumption that we have an associative algebra and so the Jacobi relations hold. In particular, we have the following identity.

\begin{Lemma} \label{DDlemma}
In the Racah algebra $R(n)$, the following relation holds:
\begin{equation}
[D_{ijk}, D_{jk\ell}] = P_{jk}\left(D_{ki\ell} +D_{i\ell j}\right).\label{DD}
\end{equation}
\end{Lemma}
\begin{proof}
 We use Lemma \ref{permutationsymofDs} to write
$ [D_{ijk}, D_{jk\ell}] =\frac12 [D_{ijk}, [ P_{k\ell}, P_{\ell j}]]$
using the fact that our algebra is associative and so the Jacobi identity holds gives
\begin{align*}
 [D_{ijk}, D_{jk\ell}] &= \tfrac12[ P_{\ell j}, [P_{k\ell}, D_{ijk}]]- \tfrac12[P_{k\ell}, [ P_{\ell j}, D_{ijk}]] \\
&= \tfrac12 [ P_{\ell j}, [P_{\ell k }, D_{kij}]]-\tfrac12 [P_{k\ell}, [ P_{\ell j}, D_{jki}]] \\
&= \tfrac12[ P_{\ell j}, P_{\ell j} P_{ki} -P_{kj}P_{\ell i}] - \tfrac12 [P_{k\ell}, P_{\ell i}P_{jk}- P_{j i}P_{\ell k}]\\
&= -\tfrac12[ P_{\ell j}, P_{jk}P_{\ell i}] - \frac12 [P_{k\ell}, P_{jk}P_{\ell i}] \\
&= -D_{\ell j k}P_{\ell i}- P_{jk} D_{j \ell i} - D_{\ell k j} P_{\ell i} -P_{jk} D_{k \ell i}
= -P_{jk} ( D_{j \ell i} +D_{k\ell i}).
\end{align*}
We can remove the negative signs using the skew symmetry of the $D$'s to write the expression~as
\begin{gather*}
[D_{ijk}, D_{jk\ell}] = P_{jk}\left( D_{ji\ell} +D_{k i\ell }\right).
\tag*{\qed}
\end{gather*} \renewcommand{\qed}{}
\end{proof}

Note that we can also apply Lemma \ref{PDcomm} to reverse the order of the right-hand side, giving the alternate form
\begin{equation}
[D_{ijk}, D_{jk\ell}] = \left(D_{ki\ell} +D_{ji\ell }\right)P_{jk}. \label{DDop}
\end{equation}

Of course, there are two other expressions that we could have used for the definition of the~$D_{jk\ell}$. The fact that each of these is equal is equivalent to verifying the Jacobi identity triples with triples of the form $[D, [P, P]]$. We will return to problem in more generality but for now we consider the following:
\[
[ D_{ijk}, [P_{jk}, P_{kl}]+ [ P_{jk}, [ P_{kl}, D_{ijk}]]+[P_{kl}, [ D_{ijk}, P_{jk}]]=0 .
\]
The first terms can be evaluated using the symmetry of the $D$'s and expression \eqref{DD}, to give
\[
[ D_{ijk}, [P_{jk}, P_{kl}]]= 2P_{jk}( D_{ji\ell}+D_{ki\ell}).
\]
The second term can be computed using \eqref{outter} as
\[ [ P_{jk}, [ P_{\ell k}, D_{kij}]] = [P_{jk}, P_{\ell j}P_{ki} - P_{kj}P_{\ell i}] = 2P_{\ell j} D_{jki}+2D_{kj\ell} P_{ik}.\]
The third term uses \eqref{intquad} and becomes
\begin{align*}
-[P_{k\ell}, [ P_{jk}, D_{ijk}]] &= -[P_{k\ell},( P_{jk}-2P_j)P_{ik} - P_{ij}(P_{jk}-2P_{k})]\\
&= -2D_{\ell kj}P_{ik} -2(P_{jk} +2P_j)D_{\ell k i}+2P_{ij}D_{\ell k j}.
\end{align*}
Putting it all together (each line is one term) and ordering $i< j< k <\ell$ gives
\begin{gather*}
0 =  -2P_{jk}(D_{ ij \ell} +D_{ i k\ell})+2P_{\ell j} D_{jki}-2D_{jk\ell} P_{ik}
+2D_{jk\ell}P_{ik}\\
\hphantom{0=} +2(P_{jk} -2P_j)D_{ik \ell}-2P_{ij}D_{jk\ell} .
\end{gather*}
 or, with cancellations,
\[ 0 =-2P_{jk}D_{ ij \ell}+2P_{\ell j} D_{jki}-4P_jD_{ik \ell}-2P_{ij}D_{jk\ell} .\]
Using the symmetry of the indices, and canceling out a factor of $-2$ gives
\[
 0 =P_{ij}D_{jk\ell}+2P_jD_{i k \ell} +P_{kj}D_{ j \ell i}+P_{\ell j} D_{jik}.
\]
We can use the permutation symmetry in the indices to express this instead as a sum over the~$i$ index,
\begin{equation}
 0 =2P_{i}D_{jk\ell}+P_{ji}D_{i k \ell} +P_{ki}D_{ i \ell j}+P_{\ell i} D_{ijk}\label{PDi}.
\end{equation}
Finally, we note that we can use \eqref{switchPD} from Lemma \ref{allPDrel|} to express the relation with the order switched,
\[
0 =2D_{jk\ell}P_{i}+D_{i k \ell}P_{ji} +D_{ i \ell j}P_{ki}+ D_{ijk}P_{\ell i}.
\]

Thus we see that this identity can be understood as a consequence of the Jacobi identity for the commutations relations for the generators, as was first shown in~\cite{CGPAV22}. We record the previous results in the following theorem and show that the remaining Jacobi relations are satisfied without additional constraints in the rank-2 case, $R(4)$. In particular, we mention that Theorem~\ref{bigthm} shows that our definition of the Racah algebra is equivalent to the definition given in \cite{CGPAV22}.

\begin{Theorem}
\label{bigthm} The Racah algebra $R(4)$ is generated by the elements $P_i$, $P_{ij}$ \eqref{P's} and $D_{ijk}$ where the $P$'s are symmetric under interchange of the indices and the $D$'s are anti-symmetric. The generators satisfy the following commutation relations:
\begin{gather}
[P_{ij}, P_{jk}]=2 D_{ijk}, \label{Ddef}\\
[P_{jk},D_{ijk}]= \left(P_{jk}-2P_j\right)P_{ki}-P_{ij}\left(P_{jk}-2P_{k}\right),\label{innerT}\\
[P_{ij}, D_{jk\ell}] =P_{i\ell}P_{jk}-P_{j\ell}P_{ik},\label{outerT}\\
 [D_{ijk}, D_{jk\ell}] = P_{jk}\left(D_{ji\ell} +D_{i\ell k}\right), \label{DDT}\\
P_{i\ell}D_{\ell jk}+ P_{j\ell} D_{\ell ki}+P_{k\ell}D_{\ell ij}+2P_{\ell}D_{ijk}=0.\label{PDT}
\end{gather}
Furthermore, these relations satisfy the Jacobi identity without additional constraints.
\end{Theorem}
\begin{proof}
 From the definition of the $P_{ij}$ and $P_{i}$ \eqref{P's} and using the algebra decomposition relations \eqref{R4decom}, it is clear that the $P$'s generate the algebra and are symmetric under interchange of indices. Lemma \ref{permutationsymofDs} shows that the $D_{ijk}$ are antisymmetric under interchange of indices.

Commutation relation \eqref{Ddef} is the definition of the $D$'s while relation \eqref{innerT} follows directly from the change of basis applied to \eqref{inner}. Lemma \ref{outerlemma} gives relation \eqref{outerT} and Lemma~\ref{DDlemma} gives~\eqref{DDT}. Finally, the identity \eqref{PDT} follows from the discussion above under a cyclic permutation, or equivalently using the Jacobi identity with $[D_{jk\ell}, [P_{ik}, P_{k\ell}]]$.

It remains to consider Jacobi identities more generally. First, we begin with the case of three~$P$'s. The Jacobi identity for triples $\{P_{jk},P_{jk},P_{ij}\} $ and $\{ P_{ki}, P_{ij}, P_{jk}\}$ are verified using Lemma \ref{JacobiC}. The case $\{P_{jk},P_{jk},P_{i\ell}\}$ is trivial since everything commutes. The final case is a~triple of the form $\{P_{i\ell},P_{jk},P_{ij}\}$ as in
\begin{equation*}
[ P_{i\ell},[P_{jk},P_{ij}]] + [P_{jk}, [ P_{ij}, P_{i\ell}]] + [P_{ij}, [P_{i\ell}, P_{jk}]]= 2[P_{i\ell}, D_{kji}]+2[P_{jk}, D_{ji\ell}].
\end{equation*}
Here we have use the definition of the $D$'s and the fact that $P$'s with disjoint indices commute. Lemma \ref{allPDrel|} gives
$ [P_{\ell i }, D_{i kj}]= [P_{kj}, D_{j \ell i}]$
and so the remaining terms cancel.

For triples containing two $P$'s and a $D$, a first choice is for all three to belong to an $R(3)$ sub-algebra. This case was addressed in the previous section, Lemma~\ref{JacobiRank1}.

For the other two cases, we assume that the $P$'s are not identical and have all 4 indices contained in them. Without loss of generality, we take $D_{ijk}$ and then either $P_{i\ell}$ and $P_{k \ell}$ or $P_{i \ell}$ and $P_{jk}$. The first case was discussed above where the identity \eqref{PDT} was derived.

For the second case, we compute
\begin{gather*}
  [ P_{i\ell}, [ P_{jk} , D_{ijk}]] + [P_{jk}, [ D_{ijk}, P_{i\ell}]]+ [ D_{ijk}, [P_{i\ell}, P_{jk}]]
 = [ P_{i\ell}, [ P_{jk} , D_{ijk}]] - [ P_{jk}, [ P_{\ell i} , D_{ijk}]].
\end{gather*}
The first term is computed using \eqref{innerT}
\begin{align*}
[ P_{i\ell}, [ P_{jk} , D_{ijk}]]& = [ P_{i\ell},  (P_{jk}-2P_j )P_{ik} - P_{ij} ( P_{jk}-2P_k )] \\
& = 2 (P_{jk}-2P_j )D_{\ell i k} - 2 D_{\ell i j} ( P_{jk}-2P_k ).
\end{align*}
Using \eqref{outerT}, the second term becomes
\begin{align*}
- [ P_{jk}, [ P_{\ell i} , D_{ijk}]] &= - [P_{jk}, P_{\ell k} P_{ij} - P_{ik} P_{\ell j}] \\
&= - 2 D_{jk\ell}P_{ij} -2 P_{\ell k} D_{kji}+2P_{ik}D_{kj\ell}+2 D_{jki} P_{\ell j}.
\end{align*}
Thus, we have
\begin{gather*}
  [ P_{i\ell}, [ P_{jk} , D_{ijk}]] + [P_{jk}, [ D_{ijk}, P_{i\ell}]]+ [ D_{ijk}, [P_{i\ell}, P_{jk}]]\\
\qquad= 2P_{ik}D_{kj\ell} +2P_{jk}D_{k \ell i} +4P_kD_{ij \ell} + 2P_{\ell k} D_{k ij} -2 D_{jk \ell}P_{ij}\\
\phantom{\qquad=}{}- 4 D_{k \ell i}P_j -2D_{j\ell i } P_{ kj}-2D_{j i k}P_{\ell j}.
\end{gather*}
These resulting two lines are equivalent to \eqref{PDT} and the equation in reverse order \eqref{switchPD} and hence add to 0.

Thus, we have shown that the Jacobi relations for two $P$'s and a $D$ are satisfied without additional constraints.
The proof of higher-order relations follow similarly and are included in Appendix \ref{AppendixA}.
\end{proof}

Let us finish the section by collecting the results above that hold in $R(n)$ and compute the last commutation relation from the definition in~\cite{CGPAV22}, thus showing that our two definitions agree. We conjecture that in this case as well the Jacobi identity is satisfied without additional constraints, though we leave this for future investigation.

\begin{Theorem}
\label{bigthmn2} The Racah algebra $R(n)$ is generated by the elements $P_i$, $P_{ij}$ \eqref{P's} and $D_{ijk}$ where the $P$'s are symmetric under interchange of the indices and the $D$'s are anti-symmetric. The generators satisfy the following commutation relations $($for distinct indices$)$:
\begin{gather}
[P_{ij}, P_{jk}]=2 D_{ijk}, \label{Ddefn}\\
[P_{jk},D_{ijk}]=  (P_{jk}-2P_j )P_{ki}-P_{ij} (P_{jk}-2P_{k} ),\\
[P_{ij}, D_{jk\ell}] =P_{i\ell}P_{jk}-P_{j\ell}P_{ik},\\
[D_{ijk}, D_{jk\ell}] = P_{jk} (D_{ji\ell} +D_{i\ell k} ),\label{DDTn}\\
[D_{ijk}, D_{k\ell m }] = P_{jk}D_{\ell mi} - P_{ki}D_{j \ell m},\label{DDnT}\\
[D_{ijk}, D_{\ell m n}] = 0,\label{DD0T}\\
P_{i\ell}D_{\ell jk}+ P_{j\ell} D_{\ell ki}+P_{k\ell}D_{\ell ij}+2P_{\ell}D_{ijk}=0.\label{PDTn}
\end{gather}
\end{Theorem}
\begin{proof}
Relations \eqref{Ddefn}--\eqref{DDTn} and \eqref{PDTn} were shown for arbitrary indices. It remains to check the final relations \eqref{DDnT} and \eqref{DD0T} which have no analog when $n=4$.

For \eqref{DDnT}, we begin with the identity
\[ [D_{ijk}, [P_{k\ell}, P_{\ell m} ]] + [P_{k\ell}, [ P_{\ell m}, D_{ijk}]] + [P_{\ell m},[D_{ijk}, P_{k\ell}]=0\]
to obtain
\begin{align*}
  [D_{ijk}, D_{k\ell m}]= [P_{\ell m}, [P_{\ell k}, D_{kij}]]
  =[ P_{\ell m}, P_{ik} P_{\ell j} - P_{jk}P_{i\ell}]
  = P_{jk}D_{\ell mi} - P_{ki}D_{j \ell m}.
\end{align*}
Similarly, \eqref{DD0T} is obtained from
\[ [D_{ijk}, [P_{\ell m }, P_{ mn} ]] + [P_{\ell m}, [ P_{ m n}, D_{ijk}]] + [P_{m n},[D_{ijk}, P_{\ell m}]]=0.\tag*{\qed}
\]\renewcommand{\qed}{}
\end{proof}

\section{Generators, presentations, and symmetries}
\label{gen}
As is apparent from the decomposition relations \eqref{R3decom} in the rank-1 case and \eqref{R4decom} for rank-2, the generators discussed above are not linearly independent. In this section, we will discuss possible choices of linear independent generators, in preparation for the following discussion of representations of the algebras. Different choices of basis will break permutation symmetry of the full algebra and we will briefly discuss what is left of the symmetry of the chosen basis.

\subsection{The rank-1 Racah algebra}
In the rank-1 case, there are 4 central elements $C_1$, $C_2$, $C_3$ and $C_{123}$. If we begin with those elements, then any choice of two of the remaining generators will be linearly independent and the third is obtained using \eqref{R3decom}. Any choice is equivalent up to permutation symmetry. A~common choice is $C_{12}$, $C_{23}$ as described in Figure~\ref{fig:pentagonC-1}. The symmetry of the generators can be realized by the interchange $1\longleftrightarrow 3$.

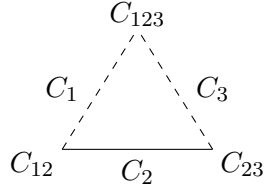
\begin{figure}[t]
\centering
\begin{tikzpicture}
\draw[black] (-1.4, -0.2) node{$C_{12}$};
\draw[black] (1.4,-0.2) node{$C_{23}$};
\draw[black] (0, 1.8) node{$C_{123}$};

\draw[black] (-1, 0.8) node{$C_1$};
\draw[black] (0, -0.3) node{$C_2$};
\draw[black] (1, 0.8) node{$C_3$};

\draw (-1, 0) -- (1, 0);
\draw[dashed] (-1, 0) -- (0,1.6);
\draw[dashed] (1, 0) -- (0,1.6);
\end{tikzpicture}
\caption{$\mathcal{D}_2$ (or $\mathcal{P}_2$, the permutation group of two elements) symmetry illustration, i.e., the dihedral group of a line. The dashed lines link two generators that commute, which solid lines link two generators that do not commute.}
\label{fig:pentagonC-1}
\end{figure}

\subsection{The rank-2 Racah algebra}
The rank-2 Racah algebra $R(4)$ is generated by 15 elements, that is 4 elements with one index~$(C_i)$, 6 elements with two indices $(C_{ij})$, 4 elements with three indices $(C_{ijk})$, and one maximal element $C_{1234}$. However, not all of these are linearly independent because of equation~\eqref{R4decom}. In fact, it is possible to reduce the set of generators to 10 elements. It is convenient to keep the 5 central terms $C_i$ and $C_{1234}$, but various choices can be made for the non-central terms. Many authors are keeping 5 two-indices elements, which possesses a dihedral-group symmetry ($\mathcal{D}_4$, isomorphic to the symmetries of a square). However, let us consider the contiguous basis, that is all the elements with continuous indices,
\[
C_1,C_2,C_3,C_4,\qquad C_{12},C_{23},C_{34},\qquad C_{123},C_{234},\qquad C_{1234}.
\]
One nice property of this choice is that a dihedral group symmetry $\mathcal{D}_5$ appears naturally. One can consider the diagram (see Figure~\ref{fig:pentagonC}).

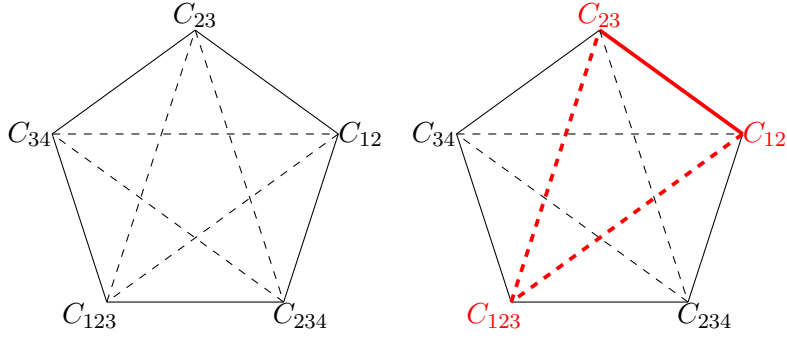
\begin{figure}[t]
\centering
\begin{tikzpicture}
\draw[black] (0, 2.2) node{$C_{23}$};
\draw[black] (-2.2, 0.6180339888) node{$C_{34}$};
\draw[black] (-1.4, -1.8) node{$C_{123}$};
\draw[black] (1.4, -1.8) node{$C_{234}$};
\draw[black] (2.2, 0.6180339888) node{$C_{12}$};
\draw (0, 2) -- (1.902113033, 0.6180339888);
\draw (1.902113033, 0.6180339888) -- (1.175570504, -1.618033989);
\draw (1.175570504, -1.618033989) -- (-1.175570504, -1.618033989);
\draw (-1.175570504, -1.618033989) -- (-1.902113033, 0.6180339888);
\draw (-1.902113033, 0.6180339888) -- (0,2);
\draw[dashed] (0, 2) -- (1.175570504, -1.618033989);
\draw[dashed] (1.175570504, -1.618033989) -- (-1.902113033, 0.6180339888);
\draw[dashed] (-1.902113033, 0.6180339888) -- (1.902113033, 0.6180339888);
\draw[dashed] (1.902113033, 0.6180339888) -- (-1.175570504, -1.618033989);
\draw[dashed] (-1.175570504, -1.618033989) -- (0, 2);
\end{tikzpicture}
\begin{tikzpicture}
\draw[red] (0, 2.2) node{$C_{23}$};
\draw[black] (-2.2, 0.6180339888) node{$C_{34}$};
\draw[red] (-1.4, -1.8) node{$C_{123}$};
\draw[black] (1.4, -1.8) node{$C_{234}$};
\draw[red] (2.2, 0.6180339888) node{$C_{12}$};
\draw[red,line width=0.5mm] (0, 2) -- (1.902113033, 0.6180339888);
\draw (1.902113033, 0.6180339888) -- (1.175570504, -1.618033989);
\draw (1.175570504, -1.618033989) -- (-1.175570504, -1.618033989);
\draw (-1.175570504, -1.618033989) -- (-1.902113033, 0.6180339888);
\draw (-1.902113033, 0.6180339888) -- (0,2);
\draw[dashed] (0, 2) -- (1.175570504, -1.618033989);
\draw[dashed] (1.175570504, -1.618033989) -- (-1.902113033, 0.6180339888);
\draw[dashed] (-1.902113033, 0.6180339888) -- (1.902113033, 0.6180339888);
\draw[dashed,red,line width=0.5mm] (1.902113033, 0.6180339888) -- (-1.175570504, -1.618033989);
\draw[dashed,red,line width=0.5mm] (-1.175570504, -1.618033989) -- (0, 2);
\end{tikzpicture}
\caption{$\mathcal{D}_5$ symmetry illustration. The dashed lines link two generators that commute, which solid lines link two generators that do not commute. On the right, the sub-Racah can be identified with red (and thicker) lines with the triangular shape from the previous section.}
\label{fig:pentagonC}
\end{figure}

In addition to the dihedral group $\mathcal{D}_5$, the full algebra is invariant under the permutation group~$\mathcal{P}_4$ acting on indices. While the $\mathcal{P}_4$ action also preserves the pentagon shape, the generators are not the contiguous basis anymore (unless one considers the identity transformation). It is interesting to note that if one combines $\mathcal{D}_5$ and $\mathcal{P}_4$, one obtains a larger symmetry group isomorphic to $\mathcal{P}_5$. This larger symmetry group appears in \cite{Crampe2023} and is presented differently, expressing the action of $\mathcal{P}_5$ on the icosidodecahedron. However, we are interested here in investigating the symmetries of the particular choice of generators and the expression of the algebra relations in terms of these generators only, so we do not carry the entire symmetry group. Hence, we focus on $\mathcal{D}_5$ and a natural expression of its action on the contiguous basis, the pentagon. 

From here, it is interesting to assign more abstract names to those elements and consider any element that can be obtained by a $\mathcal{D}_5$ or $\mathcal{P}_4$ transformation. Let us assign up to $\mathcal{D}_5$ and $\mathcal{P}_4$ transformation the following:
\begin{alignat*}{6}
&\omega_0=C_{1234},  \qquad&&\omega_1=C_1,\qquad&&\omega_2=C_2,\qquad&&\omega_3=C_3,\qquad&&\omega_4=C_4,&\\
& \Omega_0=C_{23}, \qquad&&
  \Omega_1=C_{34}, \qquad&&
  \Omega_2=C_{123}, \qquad&&
  \Omega_3=C_{234}, \qquad&&
  \Omega_4=C_{12}.&
\end{alignat*}
We can illustrate it as in Figure~\ref{fig:pentagon}.

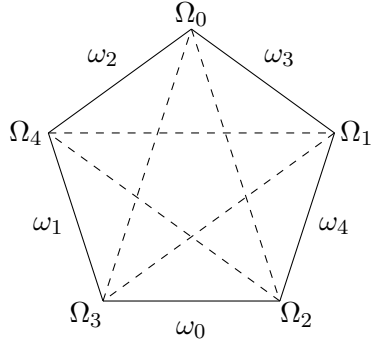
\begin{figure}[t]
\centering
\begin{tikzpicture}
\draw[black] (0, 2.2) node{$\Omega_0$};
\draw[black] (-2.2, 0.6180339888) node{$\Omega_4$};
\draw[black] (-1.4, -1.8) node{$\Omega_3$};
\draw[black] (1.4, -1.8) node{$\Omega_2$};
\draw[black] (2.2, 0.6180339888) node{$\Omega_1$};
\draw (0, 2) -- (1.902113033, 0.6180339888);
\draw (1.902113033, 0.6180339888) -- (1.175570504, -1.618033989);
\draw (1.175570504, -1.618033989) -- (-1.175570504, -1.618033989);
\draw (-1.175570504, -1.618033989) -- (-1.902113033, 0.6180339888);
\draw (-1.902113033, 0.6180339888) -- (0,2);
\draw[black] (0, -2) node{$\omega_0$};
\draw[black] (-1.902113033, -0.6180339888) node{$\omega_1$};
\draw[black] (-1.175570504, 1.618033989) node{$\omega_2$};
\draw[black] (1.175570504, 1.618033989) node{$\omega_3$};
\draw[black] (1.902113033, -0.6180339888) node{$\omega_4$};
\draw[dashed] (0, 2) -- (1.175570504, -1.618033989);
\draw[dashed] (1.175570504, -1.618033989) -- (-1.902113033, 0.6180339888);
\draw[dashed] (-1.902113033, 0.6180339888) -- (1.902113033, 0.6180339888);
\draw[dashed] (1.902113033, 0.6180339888) -- (-1.175570504, -1.618033989);
\draw[dashed] (-1.175570504, -1.618033989) -- (0, 2);
\end{tikzpicture}
\caption{$\mathcal{D}_5$ symmetry illustration.}
\label{fig:pentagon}
\end{figure}

One can also introduce new notation $\Gamma_i$ for the commutators, what we called $D_{ijk}$. They are defined as follows:
\[
\Gamma_i=\tfrac{1}{2}[\Omega_{i+2},\Omega_{i-2}],\qquad i=0,\dots,4,\qquad\mbox{mod }5.
\]
The notation may suggest that there are 5 commutators $\Gamma_i$, but they are linearly dependent,~i.e.,
$
\Gamma_0+\Gamma_1+\Gamma_2+\Gamma_3+\Gamma_4=0$.
If one consider the group action of $\mathcal{D}_5$ and $\mathcal{P}_4$, the sign of $\Gamma_i$ must be flipped depending on the group transformation. A rotation by $2n\pi/5$ from the dihedral group~$\mathcal{D}_5$ acts as follows:
\[
 \Omega_i\mapsto \Omega_{i+n}\qquad\omega_i\mapsto\omega_{i+n},\qquad \Gamma_i\mapsto \Gamma_{i+n},\qquad\mbox{mod }5,
\]
while an inversion with respect to a vertex $i$ ($\Omega_i$) and its opposite edge ($\omega_1$) acts as follows:
\[
\Omega_{i+j}\longleftrightarrow\Omega_{i-j},\qquad\omega_{i+j}\longleftrightarrow\omega_{i-j},\qquad\Gamma_{i+j}\longleftrightarrow -\Gamma_{i-j},\qquad\mbox{mod }5.
\]
Note that the mod 5 applies to the index, e.g., $\omega_6=\omega_1$.

This leads to the set of relations defining the rank-2 Racah algebra, which are invariant under the action of the groups $\mathcal{D}_5$ and $\mathcal{P}_4$.
\begin{Proposition}
All the relations between the elements of $R(4)$ can be obtained from the following relations up to the group action of $\mathcal{D}_5$.
\begin{itemize}\itemsep=0pt
  \item The definition of the $4$ commutators
  \[
    [\Omega_{i+2},\Omega_{i-2}]=2\Gamma_i,\qquad
    \sum_{i=0}^4\Gamma_i=0.
  \]
  \item The $5$ central terms
$
    [\omega_i,\omega_j]=[\omega_i,\Omega_j]=[\omega_i,\Gamma_j]=0$.
  \item If two $\Omega$s are not next to one another on the pentagon, they commute
  \begin{equation}
    [\Omega_{i-1},\Omega_{i+1}]=0.\label{OmegaGamma0}
  \end{equation}
  \item If an $\Omega$ and a $\Gamma$ share the same index, they commute $[\Omega_i,\Gamma_i]=0$.
  \item The ``inner'' equations, i.e., the equations that can be used for $R(3)$-subalgebra
  \begin{gather}
    [\Omega_i,\Gamma_{i+2}]= \Omega_i\Omega_{i+2}-\lbrace \Omega_i,\Omega_{i-1}\rbrace-\Omega_i^2+(\omega_{i+1}+\omega_{i+2}+\omega_{i+3})\Omega_i\nonumber\\
\hphantom{[\Omega_i,\Gamma_{i+2}]=}{} +(\omega_{i+2}-\omega_{i+3})\Omega_{i+2}+\omega_{i+1}(\omega_{i+3}-\omega_{i+2}).\label{inner}
  \end{gather}
  \item The ``outer'' equations
  \begin{gather}
[\Omega_i,\Gamma_{i+1}]= \sum_{k=0}^4\frac{(-1)^k}{2}\lbrace \Omega_{i+k},\Omega_{i+k-1}\rbrace+\Omega_i\Omega_{i+3}-\omega_{i+1}(\Omega_i+\Omega_{i+1})\nonumber\\
\hphantom{[\Omega_i,\Gamma_{i+1}]=}{} -\omega_{i+2}(\Omega_{i+2}+\Omega_{i+3})-\omega_{i-1}\Omega_{i-1}+\omega_{i+1}\omega_{i+2}\nonumber\\
\hphantom{[\Omega_i,\Gamma_{i+1}]=}{} +\omega_{i+1}\omega_{i-1}+\omega_{i+2}\omega_{i-1}.\label{outer}
  \end{gather}
\end{itemize}
\end{Proposition}

An important result of this section is that it gives a minimal set of equations necessary to verify a representation. In particular, there are $10$ relations coming from \eqref{inner} and $10$ from~\eqref{outer}, along with the requirement that the central terms are constant. With a choice of 8 inner equations and 2 outer equations, and some of the definitions, it is possible to satisfy all the commutation relation of $R(4)$.

Furthermore, the rank-2 Racah algebra possesses many rank-1 Racah subalgebras. Five of those subalgebras appear naturally when looking at Figure~\ref{fig:pentagonC}. The rank-1 Racah algebra made of~$C_{12}$,~$C_{23}$, $C_{123}$ and central terms has been highlighted in Figure~\ref{fig:pentagonC}, but any $\mathcal{D}_5$ action will also give a rank-1 Racah subalgebra, e.g., $\lbrace\! C_{23},C_{34},C_{324},C_2,C_3,C_4\!\rbrace$ or $\lbrace \!C_{123},C_{234},C_{23},C_{\!1},C_4,C_{1234}\!\rbrace$.

The last observation in this section is regarding Casimir operators of the rank-2 Racah algebra. The Casimir element of the rank-1 case (see equation \eqref{Casmir1}) can be written as
\begin{gather*}
\mathfrak{C}_2= \Gamma_2^2-\tfrac{1}{2}\big\lbrace \Omega_4^2,\Omega_0\big\rbrace-\tfrac{1}{2}\big\lbrace \Omega_0^2,\Omega_4\big\rbrace+\Omega_4^2+\Omega_0^2+\lbrace \Omega_4,\Omega_0\rbrace\\
\hphantom{\mathfrak{C}_2=}{}
+\tfrac{1}{2}(\omega_1+\omega_2+\omega_3+\Omega_2)(\lbrace \Omega_4,\Omega_0\rbrace-2\Omega_4-2\Omega_0)+(\omega_2-\omega_3)(\Omega_2-\omega_1)\Omega_4\\
\hphantom{\mathfrak{C}_2=}{}+(\omega_2-\omega_1)(\Omega_2-\omega_3)\Omega_0+(\omega_1+\omega_3)(\Omega_2+\omega_2)\\
\hphantom{\mathfrak{C}_2=}{}+(\omega_1\omega_3-\Omega_2\omega_2)(\Omega_2-\omega_1+\omega_2-\omega_3).
\end{gather*}
It is interesting to note that all five operators
\begin{gather*}
\mathfrak{C}_i= \Gamma_i^2-\tfrac{1}{2}\big\lbrace \Omega_{i+2}^2,\Omega_{i-2}\big\rbrace-\tfrac{1}{2}\big\lbrace \Omega_{i-2}^2,\Omega_{i+2}\big\rbrace+\Omega_{i+2}^2+\Omega_{i-2}^2+\lbrace \Omega_{i+2},\Omega_{i-2}\rbrace\\
\hphantom{\mathfrak{C}_i=}{} +\tfrac{1}{2}(\omega_{i-1}+\omega_i+\omega_{i+1}+\Omega_i)(\lbrace \Omega_{i+2},\Omega_{i-2}\rbrace-2\Omega_{i+2}-2\Omega_{i-2})\\
\hphantom{\mathfrak{C}_i=}{}+(\omega_i-\omega_{i+1})(\Omega_i-\omega_{i-1})\Omega_{i+2}+(\omega_i-\omega_{i-1})(\Omega_i-\omega_{i+1})\Omega_{i-2}\\
\hphantom{\mathfrak{C}_i=}{}+(\omega_{i-1}+\omega_{i+1})(\Omega_i+\omega_i)+(\omega_{i-1}\omega_{i+1}-\Omega_i\omega_i)(\Omega_i-\omega_{i-1}+\omega_i-\omega_{i+1}),
\end{gather*}
which can be obtained using the $\mathcal{D}_5$ symmetries, are Casimir operators. This matches the results in \cite{Crampe2023,GVZ2013,GZ88}. Each Casimir operator $\mathfrak{C}_i$, $i=0,\dots,4$ corresponds to the Casimir operator of a~rank-1 Racah subalgebra. As Cramp\'e \textit{et al.}~noted in \cite{Crampe2023}, it is quite surprising that the Casimir operator of a $R(3)$ subalgebra is also a Casimir operator of the whole algebra $R(4)$.

\section{Representations theory}
\label{repn}
\subsection{The rank-1 Racah algebra}\label{rankone}
Let us recall some of the representation theory of the rank-one universal Racah algebra. This representation theory has been studied by several authors but here we focus on the theory explained by Huang and Bockting-Conrad \cite{HBC2020}. As discussed in the previous section, we make a choice of basis using linearly independent generators $A=C_{23}$ and $B=C_{12}$. The relations are then
\begin{gather}
[A,B]=2D,\qquad
[A,D]= \lbrace A,B\rbrace+A^2-\delta A +\alpha,\label{Pres11}\\
[D,B]=\lbrace A,B\rbrace+B^2-\delta B -\beta,\label{Pres13}
\end{gather}
with $\alpha$, $\beta$ and $\delta$ central elements in the algebra. Compared to the previous section, $A$, $B$, $D$, $\alpha$, $\beta$ and $\delta$ are
\begin{alignat*}{5}
&A=C_{23},\qquad&& \alpha=(C_2-C_3)(C_1-C_{123}),\qquad&&
B=C_{12}, \qquad&& \beta=(C_1-C_2)(C_3-C_{123}),&\\
&D=D_{123},\qquad&&\delta=C_{123}+C_1+C_2+C_3.&&&&&
\end{alignat*}

The authors give an infinite-dimensional $R(3)$ module \cite[Proposition~3.1]{HBC2020} and later show that this module is isomorphic to the quotient of the universal enveloping algebra by an ideal, similar to the construction of Verma modules for Lie algebra representations.

In this case, the action of $A$ and $B$ on states $\vert j\rangle$ in the $R(3)$-module are as follows:
\begin{gather}
A\vert j\rangle=\theta_j\vert j\rangle+\vert j+1\rangle,\label{actionA}\\
B\vert j\rangle=\theta^*_j\vert j\rangle+\varphi_j\vert j-1\rangle,\label{actionB}
\end{gather}
where the raising coefficient of $A$ has been normalized to 1. In other words, $A$ acts centrally and east, while $B$ acts centrally and west, considering that the states are ordered such that $\vert j+1\rangle$ is east of $\vert j\rangle$. The coefficients $\theta_j$, $\theta_j^*$ and $\varphi_j$ can be found in the paper \cite{HBC2020} or by setting the extra parameter $s$ to zero in Appendix \ref{AppCoeff}.

The central operators $C_i$ and $C_{123}$ can be expressed as
\begin{gather*}
  C_1=c_1(c_1-1),\qquad C_2=c_2(c_2-1),\qquad C_3=c_3(c_3-1),\\
  C_{123}=(c_1+c_2+c_3+N)(c_1+c_2+c_3+N+1).
\end{gather*}

Note that in the extension to rank-2, $C_{123}$ will no longer be central. However, this Racah algebra of rank~1 will appear as a Racah subalgebra in the rank-2 case.

\subsection{The rank-2 Racah algebra}

In the remainder of the paper, we make a choice of basis for our module, the contiguous basis, breaking both the permutation symmetry of the original tensor product representation and the dihedral symmetry of the previous section. However, it is possible to use the $\mathcal{D}_5$ and $\mathcal{P}_4$ transformations to create a representation in terms of other generators.

For convenience, we choose the non-central generators as follows:
\[ \Omega_0 = C_{23}, \qquad \Omega_1=C_{34}, \qquad \Omega_2=C_{123}, \qquad \Omega_3=C_{234}, \qquad \Omega_4=C_{12},\]
and we express the central terms using constants $c_i$, i.e.,
\[
\omega_0=C_{1234}=c_0(c_0-1),\qquad \omega_i=C_i=c_i(c_i-1),\qquad i=1,\dots,4.
\]
We also have
\begin{alignat*}{4}
   &\Gamma_0=\tfrac{1}{2}[C_{123},C_{234}],\qquad && \Gamma_1=\tfrac{1}{2}[C_{234},C_{12}],\qquad &&\Gamma_2=\tfrac{1}{2}[C_{12},C_{23}],&\\
  & \Gamma_3=\tfrac{1}{2}[C_{23},C_{34}],\qquad &&\Gamma_4=\tfrac{1}{2}[C_{34},C_{123}].&&&
\end{alignat*}

Notice that $\Omega_0$ and $\Omega_4$ form an $R(3)$ subalgebra together with $\Omega_2$, so we extend the notation from Section \ref{rankone}, the rank-1 case, to construct the representation. We assume that $\Omega_0$ and~$\Omega_4$ act similarly to $A$ and $B$ in \eqref{actionA}, \eqref{actionB}, that is
$
C_{12}\vert t,s\rangle=\varphi_{t,s}\vert t-1,s\rangle+\theta_{t,s}\vert t,s\rangle$, $ C_{23}\vert t,s\rangle=\theta^*_{t,s}\vert t,s\rangle+\vert t+1,s\rangle$,
where the coefficient of the raising term of $C_{23}$ is normalized to 1. We do not make any additional assumptions on $\Omega_1$, $\Omega_2$ and $\Omega_3$ except if applied on a state, they send a state to a linear combination of itself and its 8 closest neighbours, e.g.,
\begin{gather*}
C_{123}\vert t,s\rangle = \hat{\mu}_{t,s}\vert t-1,s+1\rangle+\hat{\nu}_{t,s}\vert t,s+1\rangle+\hat{\xi}_{t,s}\vert t+1,s+1\rangle+\mu_{t,s}\vert t-1,s\rangle\\
\hphantom{C_{123}\vert t,s\rangle =}{}
+\nu_{t,s}\vert t,s\rangle+\xi_{t,s}\vert t+1,s\rangle+\check{\mu}_{t,s}\vert t-1,s-1\rangle+\check{\nu}_{t,s}\vert t,s-1\rangle+\check{\xi}_{t,s}\vert t+1,s-1\rangle.
\end{gather*}

Since $C_{12}$, $C_{23}$ and $C_{123}$ are the non-central generator of a rank-1 Racah subalgebra, one can use the presentation equations of the rank-1 Racah \eqref{Pres11}--\eqref{Pres13} to determine the coefficients. Considering that $C_{12}$ and $C_{23}$ can only move a state east-west, this implies that $C_{123}$ can only move east-west. Furthermore, by solving all the coefficients so they satisfy \eqref{Pres11}--\eqref{Pres13}, one gets that $C_{123}$ is forced to be central, that is
$
C_{123}\vert t,s\rangle =\nu_{t,s}\vert t,s\rangle$.

The remaining two generators are of the form
\begin{gather*}
C_{34}\vert t,s\rangle =\hat{\phi}^*_{t,s}\vert t-1,s+1\rangle+\hat{\vartheta}^*_{t,s}\vert t,s+1\rangle+\hat{\psi}^*_{t,s}\vert t+1,s+1\rangle+\phi^*_{t,s}\vert t-1,s\rangle+\vartheta^*_{t,s}\vert t,s\rangle\\
\phantom{C_{34}\vert t,s\rangle =}{}+\psi^*_{t,s}\vert t+1,s\rangle
+\check{\phi}^*_{t,s}\vert t-1,s-1\rangle+\check{\vartheta}^*_{t,s}\vert t,s-1\rangle+\check{\psi}^*_{t,s}\vert t+1,s-1\rangle,\\
C_{234}\vert t,s\rangle =\hat{\phi}_{t,s}\vert t-1,s+1\rangle+\hat{\vartheta}_{t,s}\vert t,s+1\rangle+\hat{\psi}_{t,s}\vert t+1,s+1\rangle+\phi_{t,s}\vert t-1,s\rangle+\vartheta_{t,s}\vert t,s\rangle\\
\phantom{C_{234}\vert t,s\rangle =}{}+\psi_{t,s}\vert t+1,s\rangle
+\check{\phi}_{t,s}\vert t-1,s-1\rangle+\check{\vartheta}_{t,s}\vert t,s-1\rangle+\check{\psi}_{t,s}\vert t+1,s-1\rangle.
\end{gather*}
Using the commutation relations \eqref{OmegaGamma0}, it is possible to show that the coefficients $\phi_{t,s}$, \smash{$\hat{\phi}_{t,s}$}, \smash{$\check{\phi}_{t,s}$} and $\psi^*_{t,s}$, \smash{$\hat{\psi}^*_{t,s}$}, \smash{$\check{\psi}^*_{t,s}$} are zero. Hence, we have
\begin{gather*}
C_{34}\vert t,s\rangle =\hat{\phi}^*_{t,s}\vert t-1,s+1\rangle+\hat{\vartheta}^*_{t,s}\vert t,s+1\rangle
+\phi^*_{t,s}\vert t-1,s\rangle+\vartheta^*_{t,s}\vert t,s\rangle\\
\phantom{C_{34}\vert t,s\rangle =}{}+\check{\phi}^*_{t,s}\vert t-1,s-1\rangle+\check{\vartheta}^*_{t,s}\vert t,s-1\rangle,\\
C_{234}\vert t,s\rangle =\hat{\vartheta}_{t,s}\vert t,s+1\rangle+\hat{\psi}_{t,s}\vert t+1,s+1\rangle
+\vartheta_{t,s}\vert t,s\rangle+\psi_{t,s}\vert t+1,s\rangle\\
\phantom{C_{234}\vert t,s\rangle =}{}+\check{\vartheta}_{t,s}\vert t,s-1\rangle+\check{\psi}_{t,s}\vert t+1,s-1\rangle.
\end{gather*}
It is interesting to note that $C_{12}$ and $C_{34}$ both act west and centrally (in combination with north-south) and they commute, while $C_{23}$ and $C_{234}$ both act east and centrally (in combination with north-south) and they also commute.

From here, we will make some assumptions on the module. We will assume that the action of any generator is zero below the line $t=0$, or in other words all states $|{-}t,s\rangle=0$ for $t>0$. In addition, we assume that any state above the line $t=s$ vanishes, that is $\vert t,t+n\rangle=0$ for~${n>0}$. Those two assumptions imply that we have an infinite triangular lattice of states.

Solving all the recurrence equations is quite a task, so we will not show all the calculations. However, it is convenient to first solve the Racah subalgebra involving $C_{12},C_{23}$ and $C_{123}$, i.e., using the commuting and the two inner relations, and then continue with the other ones. The coefficients can be found in Appendix~\ref{AppCoeff}. To illustrate the direction of the action and the meaning of each coefficient, we produced the following diagram for each generators.
The other elements of the rank-2 Racah algebra can be constructed using equation \eqref{R4decom}.

\begin{figure}[h!]
  \centering
\begin{minipage}{0.3\textwidth}
\begin{center}
$C_{12}$\\ \vspace{2mm}
~\\
\begin{tikzpicture}
\filldraw (1.1,0) circle (1pt);
\draw[-stealth] (1.1,0,0) -- (0.1,0,0);
\node at (0,-0.5){$\varphi_{t}$};
\node at (1.1,-0.5){$\theta_{t}$};
\end{tikzpicture}\\ \vspace{2mm}
west
\end{center}
\end{minipage}
\begin{minipage}{0.3\textwidth}
\begin{center}
$C_{23}$\\ \vspace{2mm}
~\\
\begin{tikzpicture}
\filldraw (0,0) circle (1pt);

\draw[-stealth] (0,0) -- (1.1,0);
\node at (0,-0.5){$\theta^*_{t}$};
\node at (1.1,-0.5){1};
\end{tikzpicture}\\ \vspace{2mm}
east
\end{center}
\end{minipage}
\begin{minipage}{0.3\textwidth}
\begin{center}
$C_{123}$\\ \vspace{0.5mm}
~\\
$\circlearrowleft$\\ \vspace{0.5mm}
\hspace{0.5mm}$\nu_{s}$\\ \vspace{0.5mm}
~\\
centrally
\end{center}
\end{minipage}
\vspace{3mm}
~\\
\begin{minipage}{0.4\textwidth}
\begin{center}
$C_{34}$\\
\begin{tikzpicture}
\filldraw (1,0) circle (1pt);

\draw[-stealth] (1,0) -- (0,0);

\draw[-stealth] (1,0) -- (1,1);
\draw[-stealth] (1,0) -- (1,-1);

\draw[-stealth] (1,0) -- (0.3,0.7);
\draw[-stealth] (1,0) -- (0.3,-0.7);

\node at (-0.1,-0.9){$\check{\phi}^*_{t,s}$};
\node at (-0.3,0){$\phi^*_{t,s}$};
\node at (-0.1,0.9){$\hat{\phi}^*_{t,s}$};

\node at (1,1.3){$\hat{\vartheta}^*_{t,s}$};
\node at (1.4,0){$\vartheta^*_{t,s}$};
\node at (1,-1.3){$\check{\vartheta}^*_{t,s}$};
\end{tikzpicture}\\
left burst
\end{center}
\end{minipage}
\begin{minipage}{0.4\textwidth}
\begin{center}
$C_{234}$\\
~\\
\begin{tikzpicture}
\filldraw (0,0) circle (1pt);

\draw[-stealth] (0,0) -- (1,0);

\draw[-stealth] (0,0) -- (0,1);
\draw[-stealth] (0,0) -- (0,-1);

\draw[-stealth] (0,0) -- (0.7,0.7);
\draw[-stealth] (0,0) -- (0.7,-0.7);

\node at (0,-1.3){$\hat{\vartheta}_{t,s}$};
\node at (-0.4,0){$\vartheta_{t,s}$};
\node at (0,1.3){$\check{\vartheta}_{t,s}$};

\node at (1,1){$1$};
\node at (1.4,0){$\psi_{s}$};
\node at (1,-1){$\hat{\psi}_{s}$};
\end{tikzpicture}\\
right burst
\end{center}
\end{minipage}
\caption{Direction of the action of each non-diagonal generator with the dependency on the parameters~$t$ and $s$.}
\end{figure}
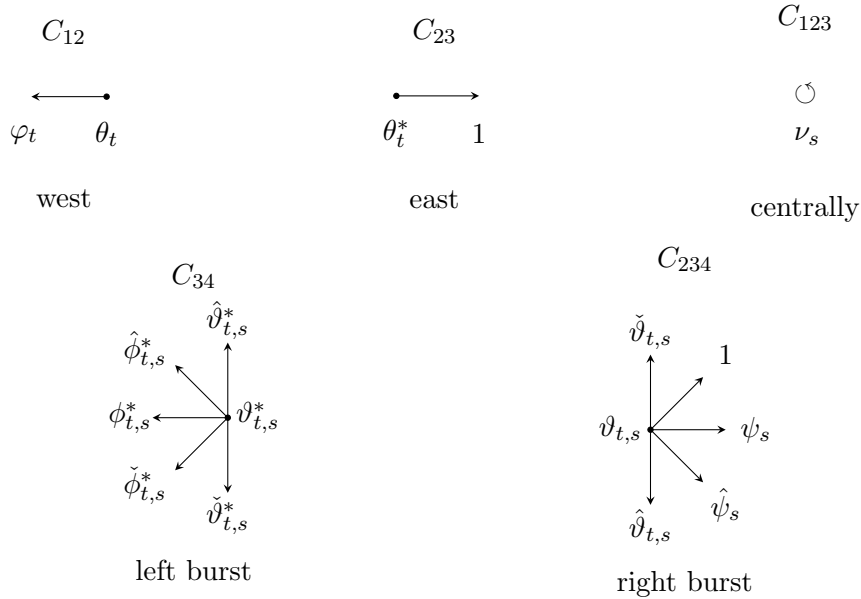

\section{Conclusion}
First, we have given a ``universal'' definition of the Racah algebra in terms of generators and relations that can be seen as an extension of the rank-1 Racah algebra using sets instead of indices. The main properties and relations generating $R(4)$ and $R(n)$ are summarized in Theorems \ref{bigthm} and \ref{bigthmn2}, respectively. Additional properties can be found throughout Section~\ref{defn}. Most significantly, we obtain the commutation relations for general $P_{ij}$ and $D_{jk\ell}$  as well as higher-order commutation relations, for example,~\eqref{DD}, given in other definitions~\cite{CGPAV22}. Similarly, we obtain the identity \eqref{PDi} first identified in \cite{CGPAV22} as arising from associativity in the algebra. We further show that in~$R(4)$ no additional relations arise from the Jacobi identities in the algebra. This result is important for our investigation into the representation theory of~$R(4)$ and in particular for future work describing Verma modules and PBW-type bases for the algebra, following on similar work as in \cite{BCH2020, HBC2020, HBC2021} and conjectures in \cite{BGVV2018}.

Next, we have explored the symmetries of some $R(4)$ generators, that is the contiguous basis. While this basis is unique, that is the set of generators is composed of operators with continuous indices, i.e., $\lbrace C_1 ,  C_2 ,  C_3 ,  C_4 ,  C_{12} ,  C_{23} ,  C_{34} ,  C_{123} ,  C_{234} ,  C_{1234}\rbrace$, many other sets of generators are isomorphic. By using the permutation symmetries $\mathcal{P}_4$ on the indices, one can get another set of generators with the same properties and commutation relations. With this symmetry, one can make a partial abstraction of the choice of generators to create a presentation and a representation. That special set of basis possesses an addition symmetry group, which is an automorphism of the set, that is the dihedral group $\mathcal{D}_5$. Indeed, the symmetry of the generators can be illustrated using a pentagon, as shown in Figure \ref{fig:pentagon}. This choice of basis leads to elegant properties. Among others, the rank-1 Racah algebra $R(3)$ using the indices $\lbrace1,2,3\rbrace$ is a subalgebra of $R(4)$, and using the $\mathcal{D}_5$ symmetry, 4 additional rank-1 Racah subalgebras arise naturally. Furthermore, if one apply $\mathcal{D}_5$ symmetries to the Casimir operator of $R(3)$, then 4 other Casimir operators arise.

Finally, we provide a novel representation for $R(4)$ using a split basis, an extension of the representation by Huang and Bockting-Conrad \cite{HBC2020}. The explicit coefficients from the action of the generators can be found in Appendix~\ref{AppCoeff}. A representation of $R(4)$ already exists \cite{Crampe2023} by Cramp\'{e}, Frappat and Ragoucy. However that representation and our representation differ. The main difference resides in the fact that the representation by Cramp\'{e} et al.\ requires that~2 generators are diagonalizable, while ours only requires one generator to be diagonalizable, which is a consequence of our choice of a split basis. In addition, Cramp\'e et al.\ assumed that their representation is from the so-called special Racah, while we did not impose this constraint. However, using our representation, we calculated the values of the Casimir operators~$\mathfrak{C}_i$,~${i=0,\dots,4}$, and we obtained that all of them vanish. Either our hypothesis led us to assume implicitly special Racah, or there are no other cases, that is the so-called special Racah is in fact Racah itself.\looseness=1

\appendix

\section{Proofs of double commutators}\label{AppendixA}
Before we begin the computations, a bit of discussion of the strategy. In each case, we try to reduce the equations to triples of $P$'s and in the case that a term contains commuting $P$'s we try to move that pair to the left. In line with this, we move inner relations to the right and outer relations to the left. It will also be useful to the following forms of the inner relation \eqref{innerT}:
\begin{gather*}
 [P_{ij}, D_{ ijk}] = P_{ij} (P_{jk}- P_{ik}) +2D_{ijk} - 2P_i P_{jk} +2P_jP_{ik},\\
[P_{ij}, D_{ ijk}] = (P_{jk}- P_{ik}) P_{ij} +2D_{ijk} - 2P_i P_{jk} +2P_jP_{ik}.\end{gather*}

\subsection[Case \{ P\_il, D\_ ijk, D\_ jkl\}]{Case $\boldsymbol{\{ P_{i \ell}, D_{ ijk}, D_{ jk\ell}\} }$}
We begin with the Jacobi relation for the set $\{ P_{i \ell}, D_{ ijk}, D_{ jk\ell}\}$
\begin{gather}
 [ P_{i \ell}, [ D_{ ijk}, D_{jkl}]]  \nonumber\\
 \qquad=  [P_{i\ell}, P_{jk}(D_{ji\ell} +D_{ki\ell}) ] \nonumber\\
 \qquad=  P_{jk}  ( P_{i\ell}( P_{j \ell}+ P_{k \ell} -P_{ij} -P_{ik}) +2D_{i \ell j} +2D_{ i\ell k}  )
   - 2P_i P_{jk}(P_{\ell j}+P_{\ell k}) \nonumber \\
   \phantom{\qquad  =}{}+ 2P_{\ell}P_{jk}(P_{ij} +P_{ik})\nonumber \\
\qquad=  P_{jk} P_{i\ell}(P_{j \ell}+P_{k\ell}-P_{ij} -P_{ik} )
  -2P_{jk}(D_{j i \ell} + D_{k i \ell})- 2 P_iP_{jk}  (P_{j \ell} +P_{k \ell} ) \nonumber \\
 \phantom{\qquad  =}{} +2P_{\ell}P_{jk}  (P_{ij} +P_{jk} ).\label{Acom}
 \end{gather}
The second term is
\begin{gather}
   [ D_{ ijk}, [D_{jk\ell}, P_{i \ell}]]\nonumber\\
   \qquad=  [[  P_{i\ell},D_{\ell jk}], D_{ijk}]  = [P_{ik} P_{\ell j} -P_{\ell k} P_{ij}, D_{ijk}]
   =  [ P_{\ell j}P_{ik} -P_{\ell k} P_{ij}, D_{ijk}]\nonumber\\
   \qquad=  [ P_{\ell j}, D_{jki}]P_{ik} - [P_{\ell k}, D_{kij}]P_{ij} + P_{\ell j} [P_{ki}, D_{kij} ]-P_{\ell k}[P_{ij}, D_{ijk}]\nonumber \\
 \qquad  =  (P_{i \ell}P_{jk}-P_{ij}P_{kl})P_{ik} - (P_{\ell j} P_{ik}- P_{jk}P_{i \ell}) P_{ij} \nonumber \\
   \phantom{\qquad  =}{} + P_{\ell j} ( P_{ki}( P_{ij}-P_{jk}) - 2D_{ijk} -2P_k P_{ij}+2P_iP_{jk} )\nonumber \\
      \phantom{\qquad  =}{}  - P_{\ell k} ( P_{ij}( P_{jk}-P_{ki}) - 2D_{ijk} -2P_i P_{jk}+2P_jP_{ki} )\nonumber \\
 \qquad =  P_{i \ell}P_{jk}(P_{ik}+P_{ij}) - (P_{ik} P_{j\ell}+P_{ij}P_{k \ell}) P_{jk} \nonumber\\
     \phantom{\qquad  =}{}   +2(P_{k \ell}-P_{j \ell})D_{ijk}+ 2P_i(P_{j\ell}+P_{k\ell}) P_{jk} - 2P_{j}P_{k\ell}P_{ik}-2P_k P_{\ell j} P_{ij}.\label{Bcom}
   \end{gather}
The third term is the same as the second, up to a sign and interchange of $i$ and $\ell$ and so reads
\begin{gather}
[ D_{jk\ell}, [ P_{i\ell}, D_{ijk}]]= - [[ P_{i\ell}, D_{ijk}], D_{\ell jk}] \nonumber\\
\hphantom{[ D_{jk\ell}, [ P_{i\ell}, D_{ijk}]]}{}
=  -P_{i \ell}P_{jk}(P_{k\ell}+P_{j\ell}) +(P_{k\ell} P_{ij}+P_{j\ell}P_{ik}) P_{jk}-2(P_{ik}-P_{ij})D_{jk\ell}\nonumber\\
\hphantom{[ D_{jk\ell}, [ P_{i\ell}, D_{ijk}]]=}{}- 2P_{\ell}(P_{ij}+P_{ik}) P_{jk} +2P_{j}P_{ik}P_{k\ell}+2P_k P_{i j} P_{j\ell}.\label{Ccom}
 \end{gather}
 Note that lines \eqref{Acom}, \eqref{Bcom} and \eqref{Ccom} sum to 0.
 Thus, finally our Jacobi identity becomes
 \begin{gather*}
  [ P_{i \ell}, [ D_{ ijk}, D_{jkl}]] +[ D_{ ijk}, [D_{jk\ell}, P_{i \ell}]] + [ D_{jk\ell}, [ P_{i\ell}, D_{ijk}]]\\
 \qquad = -2P_{jk}(D_{j i \ell} + D_{k i \ell})+2(P_{k \ell}-P_{j \ell})D_{ijk} -2(P_{ik}-P_{ij})D_{jk\ell}
 + 4P_j D_{ik\ell}+4P_k D_{ij\ell}.
 \end{gather*}
 Here we have used the fact that $[ P_{j\ell} +P_{k \ell}, P_{jk}]=0$ to cancel the $P_i$ and $P_{\ell}$ terms. Reordering and grouping gives a sum of terms which vanish due to \eqref{PDT}.

\subsection[Case \{ P\_jk, D\_ ijk, D\_ jkl\}]{Case $\boldsymbol{\{ P_{jk}, D_{ ijk}, D_{ jk\ell}\}}$}
This computation is somewhat similar to the first, although possibly more straightforward.
The first term vanishes
\begin{gather*}
 [ P_{jk}, [ D_{ ijk}, D_{jkl}]] = [P_{jk}, P_{jk}(D_{ji\ell} +D_{ki\ell}) ]
 = P_{jk}( [P_{jk},D_{k i\ell} ]+ [P_{kj}, D_{ji\ell}])=0
 \end{gather*}
due to Lemma \ref{allPDrel|}.

 The second term gives
\begin{gather*}
   [ D_{ ijk}, [D_{jk\ell}, P_{jk}]]\\
   \qquad =  [[  P_{jk},D_{ jk \ell }], D_{ijk}]
   = [ (P_{k\ell} -P_{j\ell})P_{jk} +2D_{jk\ell} -2P_jP_{k \ell}+2P_k P_{j \ell}, D_{ijk}]\nonumber \\
  \qquad = ( P_{ik}P_{j\ell}+P_{ij}P_{k \ell}-2P_{i\ell}P_{jk}) P_{jk}\nonumber \\
  \phantom{\qquad =}{}   + (P_{k \ell}-P_{j\ell})\left((P_{ik} -P_{ij})P_{jk} +2D_{ijk} -2P_jP_{ik}+2P_kP_{ij}\right) \nonumber \\
  \phantom{\qquad =}{}    - 2(D_{j i\ell}+D_{ki\ell}) P_{jk} - 2P_j[P_{\ell k},D_{kij}]+2P_k[ P_{\ell j}, D_{jki}]\nonumber\\
 \qquad =  -2P_{i\ell}P_{jk}^2 +(P_{k\ell}P_{ik}+ P_{j\ell} P_{ij}) P_{jk}
   + 2(P_{k \ell}-P_{j\ell})D_{ijk} - 2D_{i\ell j}P_{jk}- 2D_{i\ell k}P_{jk}\nonumber\\
   \phantom{\qquad =}{}    - 2P_j\left( P_{k\ell}P_{ij} -P_{i\ell}P_{jk}\right) + 2P_k\left( P_{i\ell}P_{jk} - P_{j\ell}P_{ij}\right).\label{Btrip}
   \end{gather*}

   The third term is the same as the second, up to a sign and interchange of $i$ and $\ell$ and so reads\looseness=-1
\begin{align}
[ D_{jk\ell}, [ P_{jk}, D_{ijk}]]={}&- [[ P_{jk}, D_{ijk}], D_{\ell jk}]
 = 2P_{i\ell}P_{jk}^2 -(P_{ik}P_{k\ell}+ P_{ij} P_{j\ell}) P_{jk}\nonumber \\
 &- 2(P_{ik }-P_{ij})D_{jk\ell} - 2D_{\ell i j}P_{jk}+2D_{\ell i k}P_{jk}\nonumber\\
 & +2P_j\left( P_{ik}P_{j\ell} -P_{i\ell}P_{jk}\right) - 2P_k\left( P_{i\ell}P_{jk} - P_{ij}P_{j\ell}\right). \label{Ctrip}
\end{align}
Notice that the triple of $P$'s in \eqref{Btrip} and \eqref{Ctrip} will sum together to give $2D_{\ell ki}P_{jk}+2D_{\ell ji}P_{jk}$ which will cancel with subsequent terms.

Putting everything together, we arrive at
\begin{gather*}
 [ P_{jk}, [ D_{ ijk}, D_{jkl}]] +[ D_{ ijk}, [D_{jk\ell}, P_{jk}]] + [ D_{jk\ell}, [ P_{jk}, D_{ijk}]]\\
 \qquad = 2(P_{k \ell}-P_{j\ell})D_{ijk} - 2(P_{ik }-P_{ij})D_{jk\ell} - 2D_{\ell i j}P_{jk}+2D_{\ell i k}P_{jk} +4P_j .
 \end{gather*}

\subsection[Case \{ P\_i j , D\_ ijk, D\_ jkl\}]{Case $\boldsymbol{\{ P_{i j }, D_{ ijk}, D_{ jk\ell}\}} $}
 This is the most involved case. We compute
 \begin{gather*}
 [ P_{i j}, [ D_{ ijk}, D_{jkl}]] = [P_{ij}, (D_{ji\ell} +D_{ki\ell})P_{jk} ]
= 2 (D_{ji\ell}+D_{ki\ell})D_{ijk} + [P_{ij},D_{ji \ell} +D_{ki\ell} ]P_{jk}.
\end{gather*}
 Here we have used the alternate ordering of \eqref{DDT} given in \eqref{DDop}.

 The first term on the right has not been involved in any of the previous computations as all of the previous computations have decreased the order of terms to at most triples of $P$'s, or equivalently a $P$ and $D$. We can remove this term by taking the commutator of a $PD$ relation~\eqref{PDT} (again with the opposite ordering) with a $P$ as in
 \[
   [P_{jk}, 2D_{jk\ell}P_i + D_{k\ell i}P_{ij}+ D_{ \ell ji}P_{ik}+D_{jki}P_{i\ell}]=0,
 \]
 which gives
 \begin{gather*}
  2P_i [P_{jk}, D_{jk\ell}] + [P_{jk}, D_{ k\ell i} ] P_{ij}+ [P_{jk}, D_{\ell j i}] P_{ik}+ [P_{jk}, D_{jki}]P_{i\ell} \\
  \qquad {}+ 2 D_{k\ell i}D_{kji}+2 D_{\ell j i}D_{jki}=0.
 \end{gather*}
 Thus, after a permutation of indices, we obtain
 \begin{gather}
   2(D_{ki\ell}+D_{ji\ell})D_{ijk} =
   [P_{jk}, D_{ki\ell}]P_{ij}+ [P_{kj}, D_{j\ell i}] P_{ik}\label{outerDD}\\
\hphantom{2(D_{ki\ell}+D_{ji\ell})D_{ijk} =}{} -2P_i [P_{jk},D_{jk\ell}]- [P_{jk}, D_{j k i}] P_{i\ell}.\label{innerDD}
 \end{gather}
 Note that line \eqref{innerDD} contains only inner commutation relations the \eqref{outerDD} are the outer ones. We continue this strategy in the final form of the first term in the Jacobi identity
\begin{gather*}
 [ P_{i j}, [ D_{ ijk}, D_{jkl}]] =  [P_{ji}, D_{i\ell k}]P_{jk}+[P_{jk}, D_{ki\ell}]P_{ij}+ [P_{kj}, D_{j\ell i}]P_{ik}  \\
\hphantom{[ P_{i j}, [ D_{ ijk}, D_{jkl}]] =}{} - [P_{ij}, D_{ij\ell}] P_{jk}-2P_i [P_{jk},D_{jk\ell}]- [P_{jk}, D_{jk i}]P_{i\ell}.
 \end{gather*}
 The second term of the Jacobi identity is
\begin{align*}
 [ D_{ ijk}, [D_{jk\ell}, P_{i j}]] &=  [[P_{i j},D_{jk\ell} ], D_{ijk}] = [P_{ i \ell}P_{jk}-P_{j\ell}P_{ik}, D_{ijk}]  \\
& =[ P_{\ell i}, D_{i jk}] P_{jk}-[P_{\ell j}, D_{jki}]P_{ik}
 + P_{i \ell}[P_{jk}, D_{jki}]+P_{j\ell} [P_{ik}, D_{ikj}].
 \end{align*}
 And the third
\begin{gather*}
 [ D_{jk\ell}, [P_{i j},D_{ ijk}]] =  -[[P_{i j},D_{ijk} ], D_{jk\ell}]
 = - [P_{ij}P_{jk}-P_{ik}P_{ij} -2P_iP_{jk}+2P_jP_{ik}, D_{jk\ell}]\nonumber\\
\hphantom{[ D_{jk\ell}, [P_{i j},D_{ ijk}]]}{} =  -[P_{ij}, D_{jk\ell}]P_{jk} +[P_{ik}, D_{jk\ell}]P_{ij} + P_{ik}[ P_{ij}, D_{jk\ell}] -2P_j [ P_{ik}, D_{k\ell j}]\nonumber\\
\hphantom{[ D_{jk\ell}, [P_{i j},D_{ ijk}]]=}{}- P_{ij}[P_{jk}, D_{jk\ell}]+2P_i [ P_{jk}, D_{jk\ell}].
 \end{gather*}
Thus, the Jacobi identity can be expressed as
\[ [P_{jk}, 2D_{jk\ell}P_i + D_{k\ell i}P_{ij}+ D_{ \ell ji}P_{ik}+D_{jki}P_{i\ell}] = {\rm In} +{\rm Out}, \]
with $\rm In$ and ${\rm Out}$ a collection of terms involving inner and outer quadratic relations, respectively, that is,
 \begin{gather*}
{\rm In}= - P_{ij}[P_{jk}, D_{jk\ell}]+ [P_{i \ell}[P_{jk}, D_{jki}]]+P_{j\ell} [P_{ik}, D_{ikj}]-[P_{ij}, D_{ij\ell}] P_{jk}, \\
   {\rm Out}= \left( [P_{ji}, D_{i\ell k}]+[ P_{\ell i}, D_{i jk}] -[P_{ij}, D_{jk\ell}]\right) P_{jk}
    +\left([P_{jk}, D_{ki\ell}] +[P_{ik}, D_{jk\ell}]\right)P_{ij}\\
   \phantom{ Out= }{}  + \left( [P_{kj}, D_{j\ell i}]+[P_{\ell j}, D_{jik}]\right)P_{ik}+ P_{ik}[ P_{ij}, D_{jk\ell}] -2P_j [ P_{ik}, D_{k\ell j}].
    \end{gather*}

Beginning with the inner relation terms, we can use the lower-order Jacobi identity proven earlier
$ [P_{i \ell}[P_{jk}, D_{jki}]] =[P_{jk}, [P_{i\ell}, D_{jki}]]$,
to simplify the term in ${\rm In}$,
 \begin{align}
  [P_{i \ell}[P_{jk}, D_{jki}]] ={}& [P_{jk}, P_{ij}P_{k\ell}-P_{ik}P_{j\ell}]
   = P_{jk}P_{ij}P_{k\ell}- P_{jk}P_{j\ell}P_{ik} -(P_{ij}P_{k\ell}-P_{ik}P_{j\ell})P_{jk} \nonumber \\
   ={}&2D_{kji}P_{k\ell}+2P_{ij}D_{jk\ell}- 2D_{kj\ell}P_{ik}-2P_{j\ell}D_{jki}
   \nonumber \\
   ={}& 2D_{jik}P_{k\ell} +2 D_{\ell jk}P_{ki} +2P_{ij}D_{jk\ell}+2P_{\ell j}D_{jik}.\label{innerpart1}
 \end{align}
 The other 3 terms from ${\rm In}$ give
 \begin{gather*}
   - P_{ij}[P_{jk}, D_{jk\ell}] + P_{j\ell} [P_{ik}, D_{ikj}]-[P_{ij}, D_{ij\ell}] P_{jk}\\
  \qquad  = P_{ij}( P_{j\ell}P_{jk}-P_{jk}P_{k\ell} - 2P_kP_{j\ell} +2P_j P_{k\ell}) +P_{j\ell}( P_{ik}P_{jk}-P_{ij}P_{ik}-2P_iP_{jk}+2P_kP_{ij})
\\
 \phantom{\qquad  =}{} + (P_{i\ell}P_{ij}-P_{ij}P_{j\ell} -2P_jP_{i\ell}+2P_iP_{j\ell})P_{jk}\\
   \qquad = (P_{ik}P_{j\ell} -P_{ij}P_{k\ell})P_{jk}-2P_{ij}D_{jk\ell}
     +(P_{i\ell}P_{jk}- P_{ik}P_{j\ell})P_{ij} +2P_{i\ell}D_{ijk}- 2P_{j\ell}D_{jik}\\
    \phantom{\qquad  =}{}   +4P_{k}D_{\ell j i}+2P_j(P_{ij}P_{k\ell} - P_{i\ell}P_{jk}).
 \end{gather*}
 Replacing with outer relations and permuting gives
 \begin{gather}
   - P_{ij}[P_{jk}, D_{jk\ell}]+P_{j\ell} [P_{ik}, D_{ikj}]-[P_{ij}, D_{ij\ell}] P_{jk} \nonumber \\
   \qquad=[P_{jk}, D_{ki\ell}]P_{jk}+[P_{ij}, D_{jk\ell}]P_{ij}
    -2P_{ij}D_{jk\ell} +2P_{\ell i }D_{ijk}- 2P_{j\ell}D_{jik}\nonumber\\
  \phantom{\qquad  =}{} +4P_{k}D_{\ell j i}+2P_j[P_{ik}, D_{k\ell j}].\label{innerpart2}
   \end{gather}
   Combining \eqref{innerpart1} and \eqref{innerpart2} gives
   \[
    {\rm In} = [P_{jk}, D_{ki\ell}]P_{jk}+[P_{ij}, D_{jk\ell}]P_{ij}+2P_{\ell i }D_{ijk}+2D_{\ell i k}P_{kj}+2P_j[P_{ik}, D_{k\ell j}].
    \]
   Here we have used the identity
   \[ 4P_{k}D_{\ell j i} + 2D_{jik}P_{k\ell} +2 D_{\ell jk}P_{ki}+2D_{i\ell k}P_{kj}=0.\]
Moving on to the ${\rm Out}$ term. We use Lemma \ref{allPDrel|} several times to simplify
   \begin{gather*}
   {\rm Out}=[ P_{jk}, D_{k\ell i}] P_{jk}+[P_{ij}, D_{j\ell k}]P_{ij}
 -[ P_{ij}, D_{jk\ell}]P_{ik}+ P_{ik}[ P_{ij}, D_{jk\ell}] -2P_j [ P_{ik}, D_{k\ell j}].
 \end{gather*}
    Finally, we compute
    \[
    P_{ik}[ P_{ij}, D_{jk\ell}] =P_{ik} P_{i\ell} P_{jk}- P_{ik}P_{j\ell}P_{ik}\ell= 2D_{ki\ell}P_{jk}+2P_{i\ell} D_{ikj} + [ P_{ij}, D_{jk\ell}] P_{ik},
    \]
    leaving
   \[
 {\rm Out}= [ P_{jk}, D_{k\ell i}] P_{jk}+[P_{ij}, D_{j\ell k}]P_{ij} +2D_{ki\ell}P_{jk}+2P_{i\ell} D_{ikj} -2P_j [ P_{ik}, D_{k\ell j}].
   \]
  Hence, we have shown ${\rm In} +{\rm Out} =0$ and the Jacobi identity holds for these generators without additional constraints.

 \subsection[Case \{ D\_ijl, D\_ijk, D\_jkl\}]{Case $\boldsymbol{\{ D_{ij\ell}, D_{ijk}, D_{jk\ell}\}}$}
 The final case will be for 3 distinct $D$'s. Without loss of generality we chose the indices as above and consider the quantity
 \[
  [ D_{ij\ell}, [D_{ijk}, D_{jk\ell}]] + [D_{ijk},[D_{jk\ell}, D_{ij\ell}]] +[D_{jk\ell}, [D_{ij\ell},D_{ijk}]].
  \]
 The first term becomes
 \begin{align*}
    [ D_{ij\ell}, [D_{ijk}, D_{jk\ell}]]&=  [ D_{ij\ell}, P_{jk}(D_{ji\ell}+D_{ki\ell})]
    = -[P_{kj},D_{j\ell i}] (D_{ji\ell}+D_{ki\ell})
-P_{jk}[D_{ji\ell}, D_{i\ell k} ]\\
&=  -[P_{jk}, D_{k\ell i}](D_{ij\ell}+D_{ik\ell}) -P_{jk}P_{i\ell}(D_{ijk}+D_{jk\ell}).
 \end{align*}
 A similar computation for the second term gives
 \begin{eqnarray*}
   [D_{ijk},[D_{jk\ell}, D_{ij\ell}]]=[P_{\ell j}, D_{jki}](D_{ijk}-D_{ik\ell})+P_{j\ell}P_{ik}(D_{jk\ell}-D_{ij\ell}). \end{eqnarray*}
And the third
\[ [D_{jk\ell}, [D_{ij\ell},D_{ijk}]] =[P_{ij}, D_{jk\ell}](D_{ik\ell}+D_{jk\ell})+P_{ij}P_{k\ell}(D_{ijk}+D_{ij\ell}).\]
   Combining these three terms gives the desired cancelations.

\section{Coefficients of the action of the proposed representation}\label{AppCoeff}
Below are the coefficients of the action for each generator from $R(4)$, where
\begin{gather*}
n_M=N+\sum_{i\in M}c_i,\qquad M\subseteq\lbrace1,2,3,4\rbrace,
\\
C_1\vert t,s\rangle =c_1(c_1-1)\vert t,s\rangle,
\qquad
C_2\vert t,s\rangle =c_2(c_2-1)\vert t,s\rangle,
\qquad
C_3\vert t,s\rangle =c_3(c_3-1)\vert t,s\rangle,
\\
C_4\vert t,s\rangle =c_4(c_4-1)\vert t,s\rangle,
\qquad
C_{1234}\vert t,s\rangle =n_{1234}(n_{1234}-1)\vert t,s\rangle,
\\
C_{12}\vert t,s\rangle =\varphi_{t,s}\vert t-1,s\rangle+\theta_{t,s}\vert t,s\rangle,\\
\varphi_{t,s}=(s-t)(N+1-t)(N+2c_2-t)(2n_{123}-t-s-1),\\
\theta_{t,s}=(n_{12}-t)(n_{12}-t-1),
\\
C_{23}\vert t,s\rangle=\theta^*_{t,s}\vert t,s\rangle+\vert t+1,s\rangle,\qquad
\theta^*_{t,s}=(n_{23}-t)(n_{23}-t-1),
\\
C_{123}\vert t,s\rangle =\nu_{t,s}\vert t,s\rangle,\qquad
\nu_{t,s}=(n_{123}-s)(n_{123}-s-1),
\\
C_{34}\vert t,s\rangle =\hat{\phi}^*_{t,s}\vert t-1,s+1\rangle+\hat{\vartheta}^*_{t,s}\vert t,s+1\rangle
+\phi^*_{t,s}\vert t-1,s\rangle+\vartheta^*_{t,s}\vert t,s\rangle
+\check{\phi}^*_{t,s}\vert t-1,s-1\rangle\\
\phantom{C_{34}\vert t,s\rangle =}{}+\check{\vartheta}^*_{t,s}\vert t,s-1\rangle,\\
\hat{\phi}^*_{t,s}=-(s-t)(s-t+1)(N+1-t)(N+2c_2-t),\\
\phi^*_{t,s}=\varphi_{t,s}(1-\psi_{t,s})=\left(\dfrac{1}{2}-\dfrac{n_{123}(n_{123}+2c_4-1)}{2(n_{123}-s)(n_{123}-s-1)}\right)\\
\phantom{\phi^*_{t,s}=}{}\times(s-t)(N+1-t)(N+2c_2-t)(2n_{123}-t-s-1),\\
\check{\phi}^*_{t,s}=\check{\vartheta}_{t,s}^*\dfrac{(N+1-t)(2n_{123}-t-s)(N+2c_2-t)}{2n_{12}-t-s}\\
\phantom{\check{\phi}^*_{t,s}}{}=-s\dfrac{(s+2c_4-1)(2n_{123}-s)(2n_{1234}-s-1)(2n_{123}-t-s-1)}{4(2n_{123}-2s+1)(2n_{123}-2s-1)(n_{123}-s)^2}\\
\phantom{\check{\phi}^*_{t,s}=}{} \times(N+1-t)(2n_{123}-t-s)(N+2c_2-t),\\
\hat{\vartheta}^*_{t,s}=-(s-t)(s-2c_3-t+1),\\
\vartheta^*_{t,s}=-\dfrac{n_{123}(n_{123}-t)(n_{12}-c_3-t-1)(n_{123}+2c_4-1)}{2(n_{123}-s)(n_{123}-s-1)}+\dfrac{n_{1234}(n_{1234}-1)}{2}\\
\phantom{\vartheta^*_{t,s}=}{}-\dfrac{(n_{123}-s)(n_{123}-s-1)}{2}+\dfrac{(n_{12}-t)(n_{12}-t-1)}{2}+\dfrac{c_3(c_3-1)+c_4(c_4-1)}{2},\\
\check{\vartheta}^*_{t,s}=-s\dfrac{(2n_{123}-t-s-1)(2n_{12}-s-t)(s+2c_4-1)(2n_{123}-s)(2n_{1234}-s-1)}{4(2n_{123}-2s+1)(2n_{123}-2s-1)(n_{123}-s)^2},
\\
C_{234}\vert t,s\rangle =\hat{\vartheta}_{t,s}\vert t,s+1\rangle+\hat{\psi}_{t,s}\vert t+1,s+1\rangle
  +\vartheta_{t,s}\vert t,s\rangle+\psi_{t,s}\vert t+1,s\rangle
\\
  \phantom{C_{234}\vert t,s\rangle =}{} +\check{\vartheta}_{t,s}\vert t,s-1\rangle+\check{\psi}_{t,s}\vert t+1,s-1\rangle,\\
  \hat{\vartheta}_{t,s}=(s-t)(2n_{23}-s-t-1),\\
  \vartheta_{t,s}=\dfrac{n_{123}(n_{123}-t-1)(n_{23}-c_1-t)(n_{123}+2c_4-1)}{2(n_{123}-s)(n_{123}-s-1)}+\dfrac{n_{1234}(n_{1234}-1)}{2}\\
\phantom{\vartheta_{t,s}=}{}-\dfrac{(n_{123}-s)(n_{123}-s-1)}{2}+\dfrac{(n_{23}-t)(n_{23}-t-1)}{2}+\dfrac{c_1(c_1-1)+c_4(c_4-1)}{2},\\
  \check{\vartheta}_{t,s}=\check{\psi}_{t,s}(s-2c_1-t)(2n_{123}-t-s-1)\\
  \phantom{  \check{\vartheta}_{t,s}}{}=s\dfrac{(-2c_1+s-t)(2n_{123}-t-s-1)(s+2c_4-1)(2n_{123}-s)(2n_{1234}-s-1)}{4(2n_{123}-2s+1)(2n_{123}-2s-1)(n_{123}-s)^2},\\
  \hat{\psi}_{t,s}=1,\qquad
  \psi_{t,s}=\dfrac{1}{2}+\dfrac{n_{123}(n_{123}+2c_4-1)}{2(n_{123}-s)(n_{123}-s-1)},\\
  \check{\psi}_{t,s}=\dfrac{s(s+2c_4-1)(2n_{123}-s)(2n_{1234}-s-1)}{4(2n_{123}-2s+1)(2n_{123}-2s-1)(n_{123}-s)^2}.
\end{gather*}

\subsection*{Acknowledgements}

The authors would like to thank Nicolas Cramp\'{e} and \'{E}ric Ragoucy for the fruitful discussion on the subject. The authors would also like to thank the anonymous referees for their comments and suggestions.
SP would like to acknowledge the Simons Foundation Collaboration grant \#3192112 as well as support as a CRM-Simons professorship and the hospitality of the CRM for a fruitful visit while much of the research was conducted.

\pdfbookmark[1]{References}{ref}
\LastPageEnding

\end{document}